\documentclass[12pt]{article}
\usepackage{amsmath,amssymb,amsthm}
\newtheorem{thm}{Theorem}[section]
\newtheorem{cor}[thm]{Corollary}
\newtheorem{lem}[thm]{Lemma}
\newtheorem{prop}[thm]{Proposition}
\newtheorem{conj}[thm]{Conjecture}
\newtheorem{prob}[thm]{Problem}
\theoremstyle{definition}
\newtheorem{defn}[thm]{Definition}
\newtheorem{rem}[thm]{Remark}
\numberwithin{equation}{section}

\newcommand{\To}{\Rightarrow}
\newcommand{\A}{\mathcal{A}}
\newcommand{\G}{\mathcal{G}}
\newcommand{\la}{\lambda}
\newcommand{\PP}{\mathcal{P}}
\newcommand{\SSS}{\mathcal{S}}
\newcommand{\W}{\mathcal{W}}
\newcommand{\GF}[1]{\mathrm{GF}(#1)}
\newcommand{\PG}[2]{\mathrm{PG}(#1,\,#2)}
\newcommand{\PGv}[1]{\mathrm{PG}(#1)}
\newcommand{\AG}[2]{\mathrm{AG}(#1,\,#2)}
\newcommand{\wt}{\mathrm{wt}}
\newcommand{\onevec}{\mathbf{1}}
\newcommand{\Afold}{\mathfrak{F}}

\begin{document}

\title{On the Maximality of Additive Codes}

\author{T. L. Alderson\thanks{University of New Brunswick Saint John, Saint John, NB, Canada.}}

\date{}

\maketitle

\begin{abstract}
An additive $(n,k,d)_{q^m/q}$-code is a $\GF{q}$-linear subspace of $\GF{q^m}^n$ of $\GF{q}$-dimension $km$ with minimum Hamming distance $d$. We first extend the Alderson--Bruen--Silverman (ABS) model of linear codes to the additive setting: a code of length $n$ with $q^{km}$ words over an alphabet of size $q^m$ admits an ABS model if and only if it is equivalent to a nondegenerate additive code.  We then ask whether an additive code that admits an extension must admit an \emph{additive} extension.  For linear codes ($m=1$) this is a theorem of Alderson and G\'acs.  We characterize the additive codes admitting no additive extension as those whose associated projective system of flats is complete, and we prove that the answer to the question above is again affirmative for $(n,2,d)_{9/3}$-, $(n,2,d)_{4/2}$-, and $(n,3,d)_{4/2}$-codes.  In contrast with the linear case, we show that the answer is negative in general.  Scattered linear sets yield, for each square $q$, extendable additive $(n,2,d)_{q^2/q}$-codes admitting no additive extension.  Further, a different method yields an extendable additive $(30,2,24)_{8/2}$-code with no additive extension.  Consequently, for properly additive codes, completeness of the associated projective system does not imply maximality of the code.  We conjecture that extendable $(n,2,d)_{p^2/p}$-codes, $p$ prime, always admit additive extensions.
\end{abstract}



\section{Introduction}\label{sec:intro}

For $n\geq k$, an \textit{$(n,k,d)_q$-code} $C$ is a collection of $q^k$ $n$-tuples (\textit{codewords}) over an alphabet $\A$ of size $q$, such that the minimum Hamming distance between distinct codewords of $C$ is $d$.  Thus there exist two codewords agreeing in $n-d$ coordinates, and no
two codewords agree in as many as $n-d+1$ coordinates.  
Neither linearity nor any algebraic structure on $\A$ is assumed, and $k$ need not be an integer. 

The code obtained from $C$ by deleting some fixed coordinate from every codeword is a \textit{punctured code} of $C$.  If $C$ is an $(n+1,k,d+1)_q$-code, then every punctured code of $C$ is an $(n,k,d)_q$-code; in this situation $C$ is called an \textit{extension} of the punctured code, and the punctured code is said to be \textit{extendable} to $C$.  A code admitting no extension is
\textit{maximal}.

Now let $q$ be a prime power and $m\geq 1$.  An \textit{additive $(n,k,d)_{q^m/q}$-code} is an $(n,k,d)_{q^m}$-code $C\subseteq\GF{q^m}^n$ that is moreover a $\GF{q}$-linear subspace of $\GF{q^m}^n$ (necessarily of $\GF{q}$-dimension $km$, so $km$ is an integer; we assume throughout that $k$ is also an integer). %
For $m=1$ these are exactly the linear $[n,k,d]_q$-codes.  For $m\geq 2$ the code $C$ need
not be $\GF{q^m}$-linear; additive codes that are not equivalent to linear codes are called \textit{properly additive}.  An extension of an additive code $C$ that is itself additive is an \textit{additive extension}.  We call $C$ \textit{additively maximal} if it admits no additive extension. Trivially, maximal implies additively maximal. The central question of this
paper is when the converse holds:

\begin{quote}
\textit{If an additive code admits an extension, must it admit an additive extension?}
\end{quote}

The study of additive codes was motivated largely by quantum error correction.  Calderbank, Rains, Shor and Sloane~\cite{CRSS98} showed that every binary quantum stabilizer code can be derived from a code additive over $\GF{4}$; the nonbinary case was developed by Ashikhmin and
Knill~\cite{AK01} and treated comprehensively by Ketkar, Klappenecker, Kumar and Sarvepalli~\cite{KKKS06}.  More recently the structure of additive MDS codes has received considerable attention; see Ball, Gamboa and Lavrauw~\cite{BGL22} and Adriaensen and Ball~\cite{AdBall}.  For
introductions to classical coding theory we refer to~\cite{MR0465510,MR1664228,MR2079734,1137784,MR2131191}.

Additive codes are of interest to classical coding theory since they can be strictly better than linear codes.  No linear $[21,3,18]_9$-code exists (such a code would be a maximal $\{21;3\}$-arc in $\PG29$, and maximal arcs in Desarguesian planes of odd order do not exist, by a theorem of Ball, Blokhuis and Mazzocca~\cite{BBM97}), yet an additive $(21,3,18)_{9/3}$-code does exist. An example due to Mathon (see~\cite{DDHM02}) of $21$ lines in $\PG53$, met by every hyperplane in $0$ or $3$ of them, is the projective system of such a code (see~\cite{AldersonBall26} for the general study of such systems).  Additivity is thus a genuine broadening of linearity.  
In the sequel, we provide additive codes that cannot be lengthened to an additive code, yet still admit an extension
(Sections~\ref{sec:scatter} and~\ref{sec:counter}).

For linear codes the question above was answered affirmatively by Alderson and G\'acs~\cite{MR2529622}: \textit{if a linear $(n,k,d)_q$-code admits an extension, then it admits a linear extension}.  In particular a linear code admitting no linear extension is maximal, which yields a characterization of maximal linear codes as complete weighted $(n,n-d)$-arcs in $\PG{k-1}{q}$.  The proof rests on the Bruen--Silverman model of linear codes (introduced in~\cite{ALD} and developed
in~\cite{AB1,AB3}; now commonly called the Alderson--Bruen--Silverman (ABS) model) together with results on directions determined by affine point sets.

The present paper investigates the additive analogue.  Our results reveal that the linear theory does generalize in part. In Section~\ref{sec:ABS} we extend the ABS model to additive codes and
prove:

\begin{thm}[see Theorem~\ref{thm:ABS}]\label{thm:ABS_intro}
A code of length $n$ with $q^{km}$ codewords over an alphabet of size $q^m$ admits an ABS model if and only if it is equivalent to a nondegenerate additive $(n,k,d)_{q^m/q}$-code.
\end{thm}

In Section~\ref{sec:ext} we develop and utilize the ABS model to characterize additive maximality geometrically. An additive code admits no additive extension precisely when the set $\Afold$ of $(n-d)$-fold points of its dual system meets every $(km-m-1)$-flat of $\PG{km-1}{q}$, equivalently, when its projective system of $(m-1)$-flats is complete (Corollary~\ref{cor:addmax}).  In Sections~\ref{sec:93} and~\ref{sec:43} we prove the analogue of the linear extensions theorem in certain ``small'' properly additive settings (see Theorems~\ref{thm:93} and~\ref{thm:43}).

The linear analogue does \emph{not} however persist in general. We show it fails to hold in two quite different ways.  In Section~\ref{sec:scatter} we show that  over every non-prime field, \emph{scattered linear sets} (in the sense of Blokhuis and Lavrauw~\cite{MR1772204}; see also the survey~\cite{MR2684078}) provide sets of $q^2$ points of $\AG{4}{q}$ whose direction sets contain
no line of $\PG{3}{q}$, and for square $q$ these convert into explicit counterexamples (see Theorem~\ref{thm:scattercode}).


In Section~\ref{sec:prime} we examine the prime case, and conjecture
(Conjecture~\ref{conj:m2}) that every extendable additive $(n,2,d)_{p^2/p}$-code, $p$ prime, admits an additive extension.

Even over prime fields, the parallel with linear theory does not hold in general once $m\geq3$. In Section~\ref{sec:counter} we construct a counterexample (see Theorem~\ref{thm:counter}).


Thus for additive codes, additive maximality does not imply maximality, so complete projective systems of $(m-1)$-flats need not yield maximal codes.  

Section~\ref{sec:remarks} collects the remaining open problems.

\begin{rem}\label{rem:equiv}
Two codes $C$ and $C'$ of length $n$ over alphabets $\A$ and $\A'$ of equal size are \textit{equivalent} if there are a permutation $\pi$ of the coordinate positions and bijections $f_j:\A\to\A'$ $(1\le j\le n)$ such that the map $(c_1,\ldots,c_n)\mapsto (f_1(c_{\pi(1)}),\ldots,f_n(c_{\pi(n)}))$ carries $C$ onto $C'$.  This is the natural notion of isometry for unrestricted codes and is the notion used in Theorem~\ref{thm:ABS_intro}.  Note that this is broader than the monomial (or semilinear) equivalence commonly used for linear and additive
codes.  Extendability is invariant under equivalence.
\end{rem}

\section{The ABS Model of Additive Codes}
\label{sec:ABS}

\subsection*{The Geometry}

Let $\Sigma=\PG{km}{q}$ with homogeneous coordinates $(X_0,X_1,\ldots,X_{km})$ and let $\Pi$ be the hyperplane defined by $X_{km}=0$, so $\Pi\cong\PG{km-1}{q}$.  The affine complement $\AG{km}{q}=\Sigma\setminus\Pi$ carries the structure of a $km$-dimensional affine space over $\GF{q}$, and we identify its points with the vectors $\la\in\GF{q}^{km}$ via $\la\leftrightarrow(\la,1)$.  For a nonzero vector $v\in\GF{q}^{km}$ we write $[v]$ for the corresponding point $(v,0)$ of $\Pi$.  More generally, for a nonzero $\GF{q}$-subspace
$W\subseteq\GF{q}^{km}$ we write $\PGv{W}=\{[w]\,:\,w\in W\setminus\{0\}\}$ for the flat of $\Pi$ induced by $W$, so $\PGv{W}\cong\PG{\dim W-1}{q}$, and $\PGv{\GF{q}^{km}}=\Pi$.

Fix a $\GF{q}$-basis $\mathcal{B}=\{\omega_0,\ldots,\omega_{m-1}\}$ of $\GF{q^m}$ over $\GF{q}$. This determines a $\GF{q}$-linear identification $\phi_{\mathcal{B}}:\GF{q^m}\xrightarrow{\;\sim\;}\GF{q}^m$.
Let $C$ be an additive $(n,k,d)_{q^m/q}$-code with $\GF{q}$-linear encoding isomorphism $\varepsilon:\GF{q}^{km}\to C$.  We assume throughout that $C$ is \textit{nondegenerate}: no coordinate of $C$ is identically zero.  (A degenerate coordinate contributes nothing to the distance and may be deleted; none of the questions considered here is affected.)

For each $j=1,\ldots,n$, composing the encoding map with the $j$-th projection $\pi_j$ and with $\phi_{\mathcal{B}}$ gives the \textit{$j$-th coordinate map}
\[
  \varepsilon_j \;=\; \phi_{\mathcal{B}}\circ\pi_j\circ\varepsilon
    \;:\; \GF{q}^{km}\;\longrightarrow\;\GF{q}^m,
\]
represented by an $m\times km$ matrix $M_j$ over $\GF{q}$.  For an additive code the rank $r_j$ of $M_j$ may be any integer with $0\le r_j\le m$; rank $0$ means the $j$-th coordinate is degenerate, so nondegeneracy says precisely that $r_j\geq1$ for every $j$.  Following~\cite{BLP24}, we call $C$ \textit{faithful} if every coordinate map is surjective ($r_j=m$ for every $j$); equivalently, if every alphabet symbol occurs in every coordinate.  Faithfulness is invariant under equivalence, since an equivalence preserves the number ($q^{\,r_j}$) of distinct symbols occurring in each coordinate.  (An unfaithful additive code can always be converted into a faithful one of the same length and at least the same minimum distance; see~\cite[Remark~6]{BLP24}.)  We define:
\begin{itemize}
  \item the \textit{$j$-th coordinate flat}: $\sigma_j = \PGv{\mathrm{rowspace}(M_j)}\subseteq\Pi$, a flat of dimension $r_j-1\le m-1$;
  \item the \textit{$j$-th null flat}:  $\Lambda_j = \PGv{\ker M_j}\subseteq\Pi$, a flat of dimension $km-r_j-1\ge km-m-1$.
\end{itemize}
$\sigma_j$ and $\Lambda_j$ have complementary dimensions in $\Pi$, and each determines the other in that
$\ker M_j=\mathrm{rowspace}(M_j)^{\perp}$ with respect to the standard bilinear form.  The \textit{projective system} of $C$ is the multiset $\G=\{\sigma_1,\ldots,\sigma_n\}$ and the \textit{dual system} is the multiset $\Gamma=\{\Lambda_1,\ldots,\Lambda_n\}$.  Thus $C$ is faithful if and only if $\G$ consists of $(m-1)$-flats, if and only if $\Gamma$ consists of $(km-m-1)$-flats.

\begin{rem}\label{rem:BLPdict}
In the notation of~\cite{BLP24}, an additive $(n,k,d)_{q^m/q}$-code is a code of type $[n,\,km/m,\,d]_q^{\,m}$, and $\G$ coincides with the projective system $\mathcal{X}(C)$ considered there, whose members are the column spaces of the $n$ blocks of $m$ columns of an expanded generator matrix: the $j$-th block is $M_j^{T}$, with column space $\mathrm{rowspace}(M_j)$.  In particular \emph{faithful} carries the same meaning here as in~\cite{BLP24}.
\end{rem}

\begin{rem}\label{rem:welldef}
The system $\G$ (equivalently $\Gamma$) is well defined up to a projective transformation of $\Pi$: replacing $\mathcal{B}$ by another basis multiplies each $M_j$ on the left by a fixed invertible $m\times m$ matrix, which changes neither $\mathrm{rowspace}(M_j)$ nor $\ker M_j$; replacing $\varepsilon$ by $\varepsilon\circ A$ with $A\in\mathrm{GL}(km,q)$ replaces $M_j$ by $M_jA$, which shifts all  $\sigma_j$ and $\Lambda_j$ by the common collineation induced by $A$.
\end{rem}

The key combinatorial properties of $C$ may be interpreted via the dual system through the following elementary but fundamental observation.

\begin{lem}\label{lem:agree}
Two codewords $\varepsilon(\la)$ and $\varepsilon(\mu)$, with $\la\neq\mu\in\GF{q}^{km}$, agree in the $j$-th coordinate if and only if the point $[\la-\mu]\in\Pi$ lies in the null flat $\Lambda_j$.
\end{lem}
\begin{proof}
$\varepsilon(\la)_j=\varepsilon(\mu)_j$ iff $\varepsilon_j(\la-\mu)=0$ iff $\la-\mu\in\ker M_j$ iff $[\la-\mu]\in \Lambda_j$.
\end{proof}

A point $P\in\Pi$ is a \textit{$t$-fold point} of $\Gamma$ if it lies in exactly $t$ of the null flats $\Lambda_1,\ldots,\Lambda_n$ (counted with multiplicity).  By Lemma~\ref{lem:agree}, two distinct codewords agree in exactly $t$ coordinates if and only if the corresponding direction is a $t$-fold point.  Thus, the minimum distance of the code determines that \textit{every point of $\Pi$ is at most $(n-d)$-fold, and at least one point is exactly $(n-d)$-fold.}

\begin{rem}\label{rem:duality}
Dually, in terms of the projective system $\G$: for a hyperplane $H$ of $\Pi$ with defining linear form $h$, and pole $H^\perp=[h]$, one has $\sigma_j\subseteq H$ iff every row of $M_j$ is orthogonal to $h$ iff $M_jh^T=0$ iff $H^\perp\in \Lambda_j$.  Hence the weight of the codeword $\varepsilon(\la)$ equals $n$ minus the number of members of $\G$ contained in the hyperplane polar to $[\la]$, and the minimum distance condition says that every hyperplane of $\Pi$ contains at most $n-d$ members of $\G$.  For $m=1$ the $\sigma_j$ are the points of the classical projective system (the columns of a generator matrix) and $\Gamma$ is the associated system of hyperplanes, as in~\cite{MR2529622}.  For $m\ge2$, linear codes over $\GF{q^m}$ correspond to systems whose members lie in a fixed Desarguesian $(m-1)$-spread, while properly additive codes require general $(m-1)$-flats; see~\cite{BGL22,AdBall} for this point of view in the MDS setting.
\end{rem}

\subsection*{The ABS model}

\begin{defn}\label{defn:ABSmodel}
Let $C$ be a nondegenerate additive $(n,k,d)_{q^m/q}$-code.  The \textit{ABS model} of $C$ consists of the identification of the codewords $\varepsilon(\la)$ with the affine points $(\la,1)$ of $\Sigma\setminus\Pi$, together with the dual system $\Gamma=\{\Lambda_1,\ldots,\Lambda_n\}$ in $\Pi$.
\end{defn}

For each $j$ the $q^{r_j}$ cosets of $\ker M_j$ in $\GF{q}^{km}$ form a pencil of parallel affine $(km-r_j)$-flats whose common set of points at infinity is $\Lambda_j$. The cosets are the fibres of the $j$-th coordinate, so the model realizes the symbols occurring at position $j$ as a pencil of flats through $\Lambda_j$. For a faithful coordinate, the full alphabet is realized as a pencil of $q^m$ parallel $(km-m)$-flats.  When $m=1$ nondegeneracy forces faithfulness, the null flats are hyperplanes of $\Pi$, and Definition~\ref{defn:ABSmodel} is the Bruen--Silverman model of linear codes introduced in~\cite{ALD} and developed in~\cite{AB1,AB3}, now referred to as the Alderson--Bruen--Silverman (ABS) model.

We now show that the ABS model characterizes additive codes up to equivalence (in the sense of Section~\ref{sec:intro}).

\begin{defn}\label{defn:admitsABS}
A code $C$ of length $n$ over an alphabet $\A$ with $|\A|=q^m$ and $|C|=q^{km}$ \textit{admits an ABS model} if there exist a bijection $\phi:\GF{q}^{km}\to C$ and proper flats $\Lambda_1,\ldots,\Lambda_n$ of $\Pi=\PG{km-1}{q}$ such that for all $\la\neq\mu\in\GF{q}^{km}$ and all $j\in\{1,\ldots,n\}$,
\[
  \phi(\la)_j \;=\; \phi(\mu)_j  \quad\Longleftrightarrow\quad  [\la-\mu]\in \Lambda_j.
\]
\end{defn}

\begin{rem}\label{rem:dimforced}
Writing $V_j\le\GF{q}^{km}$ for the subspace with $\PGv{V_j}=\Lambda_j$, the condition of Definition~\ref{defn:admitsABS} says that $\la+V_j\mapsto\phi(\la)_j$ is a well-defined injection from the cosets of $V_j$ into $\A$, so that $q^{km-\dim V_j}\le q^m$, i.e., $\dim \Lambda_j\ge km-m-1$.  Together with properness this gives $km-m-1\le\dim \Lambda_j\le km-2$.  Moreover $q^{km-\dim V_j}$ is exactly the number of symbols occurring in the $j$-th coordinate of $C$, so the dimensions of the $\Lambda_j$ are determined by $C$.  For $m=1$ the flats $\Lambda_j$ are therefore forced to be hyperplanes of $\Pi$, and Definition~\ref{defn:admitsABS} specializes exactly to the linear ABS model.
\end{rem}

\begin{thm}\label{thm:ABS}
Let $C$ be a code of length $n$ over an alphabet $\A$ of size $q^m$, with $|C|=q^{km}$ and minimum distance $d$.  Then $C$ admits an ABS model if and only if $C$ is equivalent to a nondegenerate additive $(n,k,d)_{q^m/q}$-code.  Moreover, the flats of the model all have dimension $km-m-1$ if and only if $C$ is equivalent to a faithful additive $(n,k,d)_{q^m/q}$-code.
\end{thm}

\begin{proof}
($\Leftarrow$)  Suppose first that $C$ is equivalent to a nondegenerate additive code $C'$ with encoding map $\varepsilon$ and null flats $\Lambda_1,\ldots,\Lambda_n$; each $\Lambda_j$ is a proper flat of $\Pi$, since each coordinate map is nonzero.  Code equivalence preserves, coordinate by coordinate, the relation of coordinate agreement. Composing $\varepsilon$ with the equivalence therefore gives a bijection $\phi:\GF{q}^{km}\to C$ satisfying, after the induced permutation of the $\Lambda_j$, the condition of Definition~\ref{defn:admitsABS} by Lemma~\ref{lem:agree}.

($\Rightarrow$)  Suppose $C$ admits an ABS model with bijection $\phi:\GF{q}^{km}\to C$ and flats $\Lambda_1,\ldots,\Lambda_n$.  For each $j$ let $V_j\le\GF{q}^{km}$ be the subspace with $\PGv{V_j}=\Lambda_j$, so that $km-m\le\dim V_j\le km-1$ by Remark~\ref{rem:dimforced}.  Identifying $\GF{q}^m\cong\GF{q^m}$ by a fixed basis, choose a $\GF{q}$-linear map $\rho_j:\GF{q}^{km}\to\GF{q^m}$ with $\ker\rho_j=V_j$, and define
\[
  \varepsilon:\GF{q}^{km}\to\GF{q^m}^n,\qquad \varepsilon(\la)=\bigl(\rho_1(\la),\ldots,\rho_n(\la)\bigr),
\]
and $C'=\varepsilon(\GF{q}^{km})$.  Since each $\rho_j$ is $\GF{q}$-linear, $C'$ is a $\GF{q}$-linear subspace of $\GF{q^m}^n$, and it is nondegenerate since each $\rho_j$ is nonzero.

\smallskip
\noindent\textit{Claim 1: $\varepsilon$ is injective.}
If $\varepsilon(\la)=\varepsilon(\mu)$ with $\la\neq\mu$ then $\la-\mu\in V_j$, i.e.\ $[\la-\mu]\in \Lambda_j$, for every $j$. By the model, $\phi(\la)$ and $\phi(\mu)$ agree in every coordinate, contradicting the injectivity of $\phi$.  Hence $|C'|=q^{km}$.

\smallskip
\noindent\textit{Claim 2: $C'$ has minimum distance $d$.}
By construction, $\varepsilon(\la)_j=\varepsilon(\mu)_j$ iff $[\la-\mu]\in \Lambda_j$ iff $\phi(\la)_j=\phi(\mu)_j$.  Hence $d_H(\varepsilon(\la),\varepsilon(\mu))=d_H(\phi(\la),\phi(\mu))$ for all $\la,\mu$, and the minimum distances of $C'$ and $C$ coincide.  Thus $C'$ is a nondegenerate additive $(n,k,d)_{q^m/q}$-code.

\smallskip
\noindent\textit{Claim 3: $C$ is equivalent to $C'$.}
Fix $j$ and define $f_j$ on the image of $\rho_j$ by $f_j(\rho_j(\la))=\phi(\la)_j$.  This is well defined: if $\rho_j(\la)=\rho_j(\mu)$ and $\la\ne\mu$ then $[\la-\mu]\in \Lambda_j$, so $\phi(\la)_j=\phi(\mu)_j$.  It is injective: if $\phi(\la)_j=\phi(\mu)_j$ with $\la\neq\mu$ then $[\la-\mu]\in \Lambda_j$, so $\rho_j(\la)=\rho_j(\mu)$.  As $|\GF{q^m}|=|\A|$, the injection $f_j$ extends to a bijection $f_j:\GF{q^m}\to\A$, and the coordinatewise map $f=(f_1,\ldots,f_n)$ satisfies $f\circ\varepsilon=\phi$.  Hence $f$ carries $C'$ onto $C$, and $C$ is equivalent to $C'$.

\smallskip
The number of symbols occurring in a given coordinate is invariant under equivalence, it equals $q^{km-\dim V_j}$ in any ABS model of $C$ (Remark~\ref{rem:dimforced}) and $q^{\,r_j}$ for an additive code with coordinate ranks $r_j$.  Hence every $\Lambda_j$ has dimension $km-m-1$ if and only if all $q^m$ symbols occur in every coordinate of $C$, if and only if every (equivalently, some) additive code equivalent to $C$ is faithful.
\end{proof}

\begin{rem}
The additive code produced in the proof depends on the choice of the maps $\rho_j$ only up to equivalence. Two choices with the same kernels differ by $\GF{q}$-linear permutations of the alphabet at each coordinate.
\end{rem}

\section{Extensions, Transversals, and Additive Extensions}
\label{sec:ext}

Throughout this section $C$ denotes a nondegenerate additive $(n,k,d)_{q^m/q}$-code with ABS model in $\Sigma=\PG{km}{q}$, hyperplane at infinity $\Pi$, and dual system $\Gamma=\{\Lambda_1,\ldots,\Lambda_n\}$.  The collection of $(n-d)$-fold points of $\Gamma$ (often called ``fat points'') will be denoted by
\[
  \Afold \;=\;
  \bigl\{P\in\Pi : P \text{ is an $(n-d)$-fold point of }\Gamma\bigr\}.
\]
Since $C$ has minimum distance $d$, the set $\Afold$ is nonempty.

For a set $S$ of points of $\AG{km}{q}$, the \textit{set of directions determined by $S$} is
\[
  D(S)\;=\;\bigl\{[X-Y] : X,Y\in S,\ X\neq Y\bigr\}\;\subseteq\;\Pi.
\]

\begin{defn}\label{defn:transversal}
A set $S\subseteq\AG{km}{q}$ is a \textit{transversal} of a point set $B\subseteq\Pi$ if $D(S)\cap B=\emptyset$.
\end{defn}

Note that subsets of transversals are transversals, and we impose no cardinality condition. However, for $m=1$, transversals of a nonempty set have at most $q^{k-1}$ points, since a set of more than $q^{k-1}$ affine points determines every direction, each parallel class of lines having $q^{k-1}$ members.  

The following provides the combinatorial interpretation of extendability; cf.\ \cite[Lemma~3.1]{MR2529622}.

\begin{lem}\label{lem:partition}
An $(n,k,d)_{q^m}$-code $C$ (not necessarily additive) is extendable if and only if $C$ can be partitioned into $q^m$ (possibly empty) classes so that any two codewords in a common class differ in at least $d+1$ coordinates.
\end{lem}

\begin{proof}
Suppose $C^+$ is an $(n+1,k,d+1)_{q^m}$-extension of $C$; say the deleted coordinate is the last.  Since $d(C^+)=d+1\ge2$, distinct codewords of $C^+$ have distinct prefixes, so puncturing is a bijection $C^+\to C$. Partitioning $C$ according to the value of the last coordinate of the corresponding word of $C^+$ provides at most $q^m$ nonempty classes. Two codewords in a common class agree in the new coordinate, so they differ in
at least $d+1$ of the first $n$ coordinates.

Conversely, given such a partition $\PP=\{S_1,\ldots,S_{q^m}\}$, label the classes by the symbols of the alphabet and append to each codeword the label of its class.  Words in a common class differ in $\geq d+1$ of the first $n$ coordinates; words in distinct classes differ in the new coordinate and in at least $d$ of the first $n$.  Hence the extended code has minimum distance at least $d+1$; and a pair of codewords of $C$ at
distance exactly $d$ (which exists) lies in two distinct classes, so the extended minimum distance is exactly $d+1$.  Thus the extended code is an $(n+1,k,d+1)_{q^m}$-code.
\end{proof}

\begin{lem}\label{lem:transpartition}
$C$ is extendable if and only if $\AG{km}{q}$ can be partitioned into $q^m$ transversals of $\Afold$.
\end{lem}

\begin{proof}
In the ABS model the codewords of $C$ are the points of $\AG{km}{q}$, so partitions of $C$ correspond to partitions of $\AG{km}{q}$, and by Lemma~\ref{lem:agree} the condition ``any two codewords of a class $S$ differ in at least $d+1$ coordinates'' says precisely that no direction determined by $S$ is a fat point (an $(n-d)$-fold point), that is, that $S$ is a transversal of $\Afold$.  Now apply Lemma~\ref{lem:partition}.
\end{proof}

Since $|\AG{km}{q}|=q^{km}$, in any such partition some class contains at least $q^{km-m}$ points.  (For $m=1$ every class has exactly $q^{k-1}$ points by the pigeonhole bound noted above.  For $m\geq2$ the class sizes need not \textit{a priori} be equal.)

Additive extensions admit a clean geometric description.

\begin{prop}\label{prop:addext}
$C$ admits an additive extension if and only if some $(km-m-1)$-flat of $\Pi$ is disjoint from $\Afold$.
\end{prop}

\begin{proof}
($\Leftarrow$)  Let $\Lambda$ be a $(km-m-1)$-flat with $\Lambda\cap\Afold=\emptyset$ and choose an $m\times km$ matrix $M_{n+1}$ of rank $m$ over $\GF{q}$ with $\PGv{\ker M_{n+1}}=\Lambda$.  Let $\varepsilon':\GF{q}^{km}\to\GF{q^m}^{n+1}$, $\varepsilon'(\la)=(\varepsilon(\la),\varepsilon_{n+1}(\la))$, where $\varepsilon_{n+1}$ is the coordinate map determined by $M_{n+1}$, and let $C^+=\varepsilon'(\GF{q}^{km})$.  For $\la\neq\mu$: if $[\la-\mu]\in \Lambda$ then $[\la-\mu]\notin\Afold$, so the fold number of $[\la-\mu]$ is at most $n-d-1$ and $d_H(\varepsilon'(\la),\varepsilon'(\mu))\geq d+1$. If $[\la-\mu]\notin \Lambda$, then the two words differ in the new coordinate and in at least $d$ of the old ones.  Finally, a pair of codewords of $C$ at distance exactly $d$ has its direction in $\Afold$, hence outside of $\Lambda$, so is at distance exactly $d+1$ in $C^+$.  Thus $C^+$ is an additive $(n+1,k,d+1)_{q^m/q}$-extension of $C$.

($\Rightarrow$)  Let $C^+$ be an additive extension of $C$. As in Lemma~\ref{lem:partition}, puncturing is a bijection $C^+\to C$, and it is $\GF{q}$-linear, so $C^+$ has encoding map $\varepsilon'=(\varepsilon,\varepsilon_{n+1})$ for some $\GF{q}$-linear map $\varepsilon_{n+1}:\GF{q}^{km}\to\GF{q^m}$, say with matrix $M_{n+1}$ of rank $r$.  Set $\Lambda'=\PGv{\ker M_{n+1}}$, a flat of dimension $km-r-1$. If some $P\in \Lambda'\cap\Afold$, take $\la\ne 0$ with $[\la]=P$, then the codewords $\varepsilon'(\la)$ and $\varepsilon'(0)$ agree in $n-d$ old coordinates and in the new one, so they are at distance $(n+1)-(n-d+1)=d<d+1$, a contradiction.  Hence $\Lambda'\cap\Afold=\emptyset$, and since $\Afold\neq\emptyset$ this forces $\Lambda'\neq\Pi$, so $r\geq1$ and $\dim \Lambda'=km-r-1\geq km-m-1$, and any $(km-m-1)$-flat contained in $\Lambda'$ is disjoint from $\Afold$.
\end{proof}

\begin{rem}\label{rem:faithfulext}
The proof shows that the new coordinate of an additive extension may always be taken to be faithful (rank $m$).  The null flat of an extending coordinate of rank $r<m$ has dimension $km-r-1>km-m-1$, and any $(km-m-1)$-subflat of it is the null flat of a faithful extending coordinate.  Nothing is therefore lost in restricting attention to faithful extensions.
\end{rem}

\begin{cor}\label{cor:addmax}
For a nondegenerate additive $(n,k,d)_{q^m/q}$-code $C$ the following are equivalent:
\begin{itemize}
\item[(i)] $C$ is additively maximal;
\item[(ii)] every $(km-m-1)$-flat of $\Pi$ meets $\Afold$;
\item[(iii)] the projective system $\G$ of $C$ is complete: no $(m-1)$-flat $\sigma_{n+1}$ can be adjoined to $\G$ so that the resulting system is the projective system of an $(n+1,k,d+1)_{q^m/q}$-code.
\end{itemize}
\end{cor}

\begin{proof}
(i)$\iff$(ii) is Proposition~\ref{prop:addext}.  For (ii)$\iff$(iii), adjoining an $(m-1)$-flat $\sigma_{n+1}=\PGv{W}$ with associated null flat $\Lambda_{n+1}=\PGv{W^{\perp}}$ produces the system of an $(n+1,k,d+1)$-code precisely when no point of $\Pi$ becomes $(n-d+1)$-fold, i.e.\ precisely when $\Lambda_{n+1}\cap\Afold=\emptyset$.
\end{proof}

For $m=1$, Corollary~\ref{cor:addmax} combined with the main theorem of~\cite{MR2529622} says that a linear code is \emph{maximal} iff $\Afold$ is an intersection set (blocking set with respect to hyperplanes) of $\PG{k-1}{q}$.  Whether ``additively maximal'' can be upgraded to ``maximal'' for $m\geq2$ is precisely the question of the introduction.

We close this section with a general necessary condition for extendability, cf.\ \cite[Proposition~2]{MR1869411}.

\begin{prop}\label{prop:mflat}
If $C$ is extendable, then $\Afold$ contains no $m$-flat of $\Pi$.
\end{prop}

\begin{proof}
Let $F$ be an $m$-flat of $\Pi$ and let $T$ be a transversal of $\Afold$ as in Lemma~\ref{lem:transpartition}, with $|T|\ge q^{km-m}$.  The affine $(m+1)$-flats of $\Sigma$ whose set of infinite points is $F$ partition $\AG{km}{q}$ into $q^{km-m-1}<|T|$ classes. Two points of $T$ thus lie in a common such flat, and their direction lies in $F$, so $F\cap D(T)\neq\emptyset$. By definition of a transversal, $D(T)\cap\Afold=\emptyset$, whence $F\not\subseteq\Afold$.
\end{proof}

\section{Extendable $(n,2,d)_{9/3}$-Codes and $(n,2,d)_{4/2}$-Codes are Additively Extendable}
\label{sec:93}

In this section we prove the additive analogue of the linear extensions theorem for the smallest properly additive parameters: $(n,2,d)_{9/3}$- and $(n,2,d)_{4/2}$-codes.  The main work concerns the ternary case $q=3$, $m=2$, $k=2$.  Here $\Sigma=\PG{4}{3}$, $\Pi=\PG{3}{3}$, the null flats are lines of $\Pi$, and the fibres of a new faithful coordinate have $q^{km-m}=9$ points (cf.\ Remark~\ref{rem:faithfulext}); the binary case then follows by the same argument in simpler form.

\begin{lem}\label{lem:dir9}
Let $T$ be a set of $9$ points of $\AG{4}{3}$.  Then $D(T)$ contains a line of $\Pi=\PG{3}{3}$.
\end{lem}

\begin{proof}
Call a line of $\AG{4}{3}$ a \textit{secant} of $T$ if it contains at least two points of $T$; since an affine line over $\GF{3}$ has exactly $3$ points, a secant contains two or three points of $T$ (a \textit{$2$-secant} or a \textit{$3$-secant}).

\smallskip
\noindent\textit{Claim:}  If some plane of $\AG{4}{3}$ contains at least four points of $T$, then $D(T)$ contains a line of $\Pi$.

Let $\alpha$ be such a plane and let $\ell_\alpha$ be the line of $\Pi$ in which the projective closure of $\alpha$ meets $\Pi$.  For each point $[u]\in\ell_\alpha$, the lines of $\alpha$ with direction $[u]$ form a parallel class with exactly $3$ members.  Two of the four points thus lie on a common member, so $[u]\in D(T)$.  Hence $\ell_\alpha\subseteq D(T)$.

\smallskip
\noindent Assume now, for a contradiction, that $D(T)$ contains no line of $\Pi$.  By the Claim, no plane contains four points of $T$.  It follows that $T$ has no $3$-secant, so each of the $\binom{9}{2}=36$ pairs of points of $T$ lies on its own $2$-secant, and these $36$ secants are pairwise distinct.  Moreover no two of them have a common direction since two such secants would put four points of $T$ in a plane.  The $36$ secants therefore determine $36$ distinct directions, and $B:=\Pi\setminus D(T)$ consists of exactly $40-36=4$ points.

\smallskip
Since $D(T)$ contains no line, every line of $\Pi$ contains a point of $B$.  Choose a point $P\in\Pi\setminus B$. The $13$ lines of $\Pi$ through $P$ meet pairwise only in $P$, and each contains a point of $B$; hence $|B|\geq13$, contradicting $|B|=4$. (This is the $\PG{3}{3}$ case of the theorem of Bose and Burton~\cite{BoseBurton}: a point set meeting every line of $\PG{3}{q}$ has at least $q^2+q+1$ points.)

The contradiction shows that $D(T)$ contains a line of $\Pi$.
\end{proof}

\begin{thm}\label{thm:93}
Let $C$ be a nondegenerate additive $(n,2,d)_{9/3}$-code.  Then $C$ is extendable if and only if $C$ admits an additive extension.
\end{thm}

\begin{proof}
One implication is trivial.  For the other, suppose $C$ is extendable. By Lemma~\ref{lem:transpartition} there is a partition of $\AG{4}{3}$ into $9$ transversals of $\Afold$; one class contains at least $81/9=9$ points, and any $9$ of them form a transversal $T$ of $\Afold$ with $|T|=9$.  By Lemma~\ref{lem:dir9}, $D(T)$ contains a line $\ell$ of $\Pi$; since $D(T)\cap\Afold=\emptyset$ we get $\ell\cap\Afold=\emptyset$.  Now Proposition~\ref{prop:addext} (with $km-m-1=1$) provides an additive extension. 
\end{proof}

\begin{cor}\label{cor:93max}
A nondegenerate additive $(n,2,d)_{9/3}$-code is maximal if and only if it is additively maximal, if and only if every line of $\PG{3}{3}$ meets $\Afold$, if and only if its projective system of lines in $\PG{3}{3}$ is complete.
\end{cor}

Lemma~\ref{lem:dir9} also holds, trivially, over $\GF{2}$, by the same pigeonhole observation: if $T\subseteq\AG{4}{2}$ with $|T|=4$, then any three points of $T$ are non-collinear (an affine line over $\GF{2}$ has only two points), hence are $q+1=3$ points of the affine plane they span, whose line at infinity (a line of $\PG{3}{2}$) is therefore contained in $D(T)$.  The same argument as in Theorem~\ref{thm:93} then provides the following:

\begin{prop}\label{prop:42max}
	A nondegenerate additive $(n,2,d)_{4/2}$-code is maximal if and only if it is additively maximal, if and only if every line of $\PG{3}{2}$ meets $\Afold$, if and only if its projective system of lines in $\PG{3}{2}$ is complete.
\end{prop}


In this section we considered cases with $m=k=2$.  The corresponding direction problem for $m=2$ and $k\geq3$ concerns sets of $q^{2k-2}$ points of $\AG{2k}{q}$ and $(2k-3)$-flats of $\PG{2k-1}{q}$.  In the next section we settle the first instance, $(q,m,k)=(2,2,3)$, in the affirmative.

\section{Extendable $(n,3,d)_{4/2}$-Codes are Additively Extendable}
\label{sec:43}

In this section we take $q=2$, $m=2$, $k=3$. In this setting, additive codes are $\GF2$-linear subspaces of $\GF4^n$ with $2^6$ codewords, $\Sigma=\PG62$, $\Pi=\PG52$, null flats are solids ($3$-flats) of $\Pi$, and the fibres of a new faithful coordinate have $q^{km-m}=16$ points. The direction result required is the following analogue of Lemma~\ref{lem:dir9}, the proof of which is the main work of the section.

\begin{thm}\label{thm:dir16}
Let $T$ be a set of $16$ points of $\AG62$.  Then $D(T)$ contains a solid of $\Pi=\PG52$.
\end{thm}

Throughout the section the ground field is $\GF2$.  We identify affine spaces $\AG{r}{2}$ with $\GF2^r$, use the notation $D(\cdot)$ of Section~\ref{sec:ext} in any such space (the directions lying in the hyperplane at infinity $\PG{r-1}{2}$). Recall we identify a nonzero vector $v$ with the projective point $[v]$, so that direction sets may be read as sets of nonzero vectors.  Over $\GF2$ an affine line has exactly two points, so $D(S)=\{[\la+\mu]:\la\neq\mu\in S\}$, and $D(S)$ is invariant under translation of $S$; a $2$-flat has four points, and two distinct affine lines with a common point at infinity are disjoint, their union being a $2$-flat.

We begin with an observation regarding seven-point sets in dimension five.

\begin{lem}\label{lem:seven}
If $S\subseteq\AG52$ with $|S|=7$, then $D(S)$ contains a plane of $\PG42$.
\end{lem}

\begin{proof}
A $4$-subset $\{a,b,c,e\}$ of $\GF2^5$ is an \textit{affine plane} if $a+b+c+e=0$. Call $S$ \textit{Sidon} if it contains no affine plane, that is, if the $\binom72=21$ pairwise sums of distinct elements of $S$ are pairwise distinct.  The proof rests on the following observation.

\smallskip
\noindent$(\ast)$ \textit{If $a,b,c,e\in S$ are distinct, $a+b+c+e\neq0$, and $a+b+c+e\in D(S)$, then $D(S)$ contains the plane $\PGv{\langle a+b,\,a+c,\,a+e\rangle}$.}

\smallskip
Indeed, put $d_1=a+b$, $d_2=a+c$, $d_3=a+e$.  A single $d_i$ is nonzero since $a,b,c,e$ are distinct; a sum of two of them is one of $b+c$, $b+e$, $c+e$, again nonzero; and $d_1+d_2+d_3=a+b+c+e\neq0$ by hypothesis.  Hence $\langle d_1,d_2,d_3\rangle$ is $3$-dimensional, with nonzero vectors 
\[
d_1,\; d_2,\; d_3,\quad d_1+d_2=b+c,\; d_1+d_3=b+e,\; d_2+d_3=c+e,\quad d_1+d_2+d_3=a+b+c+e,
\]
the first six being sums of two distinct elements of $S$, hence in $D(S)$, and the seventh in $D(S)$ by hypothesis.  This proves $(\ast)$.

\smallskip
Suppose first that $S$ contains an affine plane $\{a,b,c,e\}$, so $e=a+b+c$.  Choose $t\in S\setminus\{a,b,c,e\}$ (possible as $|S|=7$). Then $a+b+c+t=e+t$ is nonzero and is a sum of two distinct elements of $S$, so $(\ast)$ applies to $a,b,c,t$.

Suppose now that $S$ is Sidon.  No $4$-subset of $S$ sums to zero, so by $(\ast)$ it suffices to find a $4$-subset of $S$ whose sum lies in $D(S)$.  Consider the sum map $\sigma(A)=\sum_{u\in A}u$ on the $\binom74=35$ four-subsets $A\subseteq S$.  Every value of $\sigma$ is nonzero, and no value is attained by three distinct $4$-subsets. Indeed, if $\sigma(A)=\sigma(B)$ with $A\neq B$, then $\sum_{u\in A\triangle B}u=0$ with $|A\triangle B|\in\{2,4,6\}$; size $2$ would force two equal elements and size $4$ an affine plane, so $|A\triangle B|=6$, whence $|A\cap B|=1$ and $A\cup B=S$. Writing $A\cap B=\{w\}$, the vanishing sum reads $\sum_{u\in S}u=w$, so $w$ is determined by $S$ alone and $B=\{w\}\cup(S\setminus A)$ is determined by $A$.  Consequently $\sigma$ takes at least $\lceil35/2\rceil=18$ distinct nonzero values.  On the other hand, since $S$ is Sidon, $D(S)$ consists of exactly $21$ of the $31$ nonzero vectors of $\GF2^5$, so only $10$ nonzero vectors lie outside $D(S)$.  Since $18>10$, some value $\sigma(A)$ lies in $D(S)$, and $(\ast)$ applies to $A$.
\end{proof}

\begin{rem}\label{rem:sevensharp}
The conclusion of Lemma~\ref{lem:seven} does not hold if the ambient dimension is raised from $5$ to $6$, so the hypothesis $S\subseteq\AG52$ cannot be relaxed.   In $\GF2^6$ a Sidon $7$-set determines only $21$ of the $63$ nonzero directions, so it avoids $63-21=42\geq18$ of them, and the final count  in the proof yields no contradiction.  This is not merely a shortcoming of the proof: the Sidon set $S=\{0,e_1,\ldots,e_6\}$ consisting of the zero vector and the standard basis of $\GF2^6$ determines precisely the points $[v]$ with $v$ of Hamming weight $1$ or $2$, and no plane of $\PG52$ consists of such points, since every $3$-dimensional subspace of $\GF2^6$ contains a vector of weight $3$ or $4$ (Lemma~\ref{lem:block34}, proved in Section~\ref{sec:counter}).  This set reappears, inside the transversal $T=\{0,e_1,\ldots,e_6,\onevec\}$ of Proposition~\ref{prop:ext}, in the counterexample of Section~\ref{sec:counter}.
\end{rem}

\begin{lem}\label{lem:BC}
Let $B\subseteq E\subseteq\AG52$ with $|B|+|E|=16$ and $6\le|B|\le8$, and let
\[
D(B,E)\;=\;\bigl\{[\beta+\gamma] : \beta\in B,\ \gamma\in E,\ \beta\neq\gamma\bigr\}.
\]
Then $D(B,E)$ contains a plane of $\PG42$.
\end{lem}

\begin{proof}
If $|B|\ge7$, then choose a $7$-subset $S\subseteq B$, so $D(S)\subseteq D(B,E)$.  %
If $|B|=6$ then $|E|=10$, and we choose $\gamma\in E\setminus B$ and put $S=B\cup\{\gamma\}$.  Each sum of two distinct elements of $S$ has at least one summand in $B$ and the other in $E$, so again $D(S)\subseteq D(B,E)$.  In either case Lemma~\ref{lem:seven} applies.
\end{proof}

\begin{proof}[Proof of Theorem~\ref{thm:dir16}]
The $\binom{16}2=120$ pairs of points of $T$ determine directions among the $63$ points of $\Pi$, so two distinct pairs share a direction and their four points form a $2$-flat $\alpha\subseteq T$.  Let $W\le\GF2^6$ be the $2$-dimensional subspace with $\alpha=v+W$ (equivalently, $W=\{a+b:a,b\in\alpha\}$), i.e.\ the translation subspace of $\alpha$, and let $\ell=\PGv{W}$ be its line at infinity.  Every nonzero $w\in W$ is a difference of two points of $\alpha\subseteq T$, so
\begin{equation}\label{eq:ellD}
\ell\subseteq D(T).
\end{equation}

\smallskip
Write $\overline V=\GF2^6/W\cong\GF2^4$ and $\la\mapsto\bar\la$ for the quotient map.  The cosets of $W$ are the $2$-flats parallel to $\alpha$. For $X\in\overline V$ let $\alpha_X$ denote the corresponding $2$-flat and $n_X=|T\cap\alpha_X|$, and let $O=\bar\alpha$, so that
\begin{equation}\label{eq:twelve}
n_O=4,\qquad \sum_{X\neq O}n_X=12 .
\end{equation}
We regard $\overline V$ as the affine space $\AG42$ with hyperplane at infinity $\overline\Pi=\PGv{\overline V}\cong\PG32$.  For a point $x=[\bar v]\in\overline\Pi$, the plane $\gamma_x=\PGv{W+\langle v\rangle}$ of $\Pi$ contains $\ell$, and the four points of $\gamma_x\setminus\ell$ are the $[v+w]$, $w\in W$.  Geometrically, $\overline V$ together with $\overline\Pi$ is the quotient of $\Sigma$ at $\ell$: the points of $\overline\Pi$ are the planes of $\Pi$ through $\ell$ (the plane corresponding to $x$ being $\gamma_x$), and the affine points are the planes of $\Sigma$ meeting $\Pi$ precisely in $\ell$, that is, the $2$-flats parallel to $\alpha$.  We shall call $x$ \textit{covered} if $\gamma_x\setminus\ell\subseteq D(T)$.

\smallskip
\noindent\textit{Claim 1:}  \textit{If all three points of some line of $\overline\Pi$ are covered, then $D(T)$ contains a solid.}\\ 
Indeed, such a line is $\PGv{\overline U}$ for a $2$-dimensional $\overline U\le\overline V$.  Its preimage $U\le\GF2^6$ is $4$-dimensional and contains $W$, and $\PGv U$ is the union of $\ell$ and the three sets $\gamma_x\setminus\ell$, $x\in\PGv{\overline U}$ ($3+3\cdot4=15$ points).  By~\eqref{eq:ellD} and the covered property, the solid $\PGv U$ lies in $D(T)$.

\smallskip
By Claim~1 we may assume for the remainder of the proof that \emph{no line of $\overline\Pi$ has all three of its points covered}.

\smallskip
\noindent\textit{Claim 2.  If $X\neq O$ and $n_X>0$, then $[O+X]$ is covered.}\\
To see this, fix $\la\in T\cap\alpha_X$.  As $\mu$ varies over the four points of $\alpha\subseteq T$, the differences $\la+\mu$ range over a coset of $W$ disjoint from $W$, namely over the four vectors representing the points of $\gamma_x\setminus\ell$ with $x=[O+X]$.  Hence $\gamma_x\setminus\ell\subseteq D(T)$.

\smallskip
Let $\SSS=\{[O+X] : X\neq O,\ n_X>0\}\subseteq\overline\Pi$ and $s=|\SSS|$.  Since $X\mapsto[O+X]$ is a bijection from $\overline V\setminus\{O\}$ to $\overline\Pi$, we have $s=\#\{X\neq O : n_X>0\}$, and by Claim~2 every point of $\SSS$ is
covered.

\smallskip
\noindent\textit{Claim 3.  $n_X+n_Y\le4$ for all distinct $X,Y\neq O$ with $n_X,n_Y\ge1$.}

Indeed, suppose $n_X+n_Y\ge5$.  The points $x=[O+X]$, $y=[O+Y]$ and $z=[X+Y]$ are distinct and collinear in $\overline\Pi$ (their representing vectors sum to $0$).  Let $[u]$ be any of the four points of $\gamma_z\setminus\ell$, so $\bar u=X+Y$.  Translation by $u$ maps $\alpha_X$ bijectively onto $\alpha_Y$. Since $n_X+n_Y>4=|\alpha_Y|$, the sets $(T\cap\alpha_X)+u$ and $T\cap\alpha_Y$ meet, giving $\la\in T\cap\alpha_X$ with $\la+u\in T\cap\alpha_Y$, whence $[u]\in D(T)$.  Thus $z$ is covered, while $x$ and $y$ are covered by Claim~2, so the line $\{x,y,z\}$ is fully covered, contrary to the standing assumption.

\smallskip
We record three consequences.
\begin{itemize}
\item[(a)] \textit{$\SSS$ is a cap of $\overline\Pi$} (no three points collinear): a line of $\overline\Pi$ with all three points in $\SSS$ would consist of covered points.  Caps of $\PG32$ have at most $8$ points (through any point of a cap pass $7$ lines, each containing at most one further cap point), so $s\le8$.
\item[(b)] \textit{$n_X\in\{1,2\}$ whenever $X\neq O$ and $n_X>0$.} Note first that $s\ge3$, by~\eqref{eq:twelve} and $n_X\le4$.  If some $n_X=4$, then any second class $Y$ with $n_Y\ge1$ violates Claim~3.  If some $n_X=3$, then Claim~3 forces $n_Y=1$ for every other class counted by $s$, so $12=3+(s-1)\le10$ by (a), a contradiction. \item[(c)] Let $\kappa=\#\{X : n_X=2\}$.  By~\eqref{eq:twelve} and (b), $2\kappa+(s-\kappa)=12$, so $\kappa=12-s\ge0$ and $s-\kappa=2s-12\ge0$; with (a), 
\begin{equation}\label{eq:srange}
6\le s\le 8,\qquad \kappa=12-s\in\{4,5,6\}.
\end{equation}
\end{itemize}

\smallskip
For each of the $\kappa$ classes $X$ with $n_X=2$, the set $L_X=T\cap\alpha_X$ is an affine line whose point at infinity lies on $\ell$.  Suppose two of them, $L_X=\{\la_1,\la_2\}$ and $L_Y=\{\mu_1,\mu_2\}$, had distinct points at infinity, i.e.\ $\la_1+\la_2\neq\mu_1+\mu_2$.  The four cross differences $\la_i+\mu_j$ are then pairwise distinct ($\la_1+\mu_1=\la_1+\mu_2$ would give $\mu_1=\mu_2$, and $\la_1+\mu_1=\la_2+\mu_2$ would give
$\la_1+\la_2=\mu_1+\mu_2$), and they represent points of $\gamma_z\setminus\ell$ for $z=[X+Y]$. Since $|\gamma_z\setminus\ell|=4$, the point $z$ is covered, and with Claim~2 the line $\{[O+X],[O+Y],z\}$ is fully covered, contrary to the standing assumption.  Hence all $\kappa$ lines $L_X$ have a common point at infinity $P=[p]\in\ell$.

Write $\widehat V=\GF2^6/\langle p\rangle\cong\GF2^5$ and $\la\mapsto\widehat\la$ for the quotient map, and regard $\widehat V$ as $\AG52$ with hyperplane at infinity $\widehat\Pi=\PGv{\widehat V}\cong\PG42$.  Call an affine line with point at infinity $P$ a \textit{$P$-line}; the fibres of the quotient over the points of $\AG52$ are exactly the $P$-lines, each $P$-line lies in a single class $\alpha_X$ (as $p\in W$), and each class
contains exactly two $P$-lines.  Geometrically, $\widehat V$ together with $\widehat\Pi$ is the quotient of $\Sigma$ at $P$, the affine points being the $P$-lines.  Put
\[
B=\{\widehat\la : \{\la,\la+p\}\subseteq T\},\qquad E=\{\widehat\la : \{\la,\la+p\}\cap T\neq\emptyset\},
\]
so $B\subseteq E\subseteq\AG52$, distinct $P$-lines giving distinct points of $\AG52$.

The two $P$-lines of $\alpha$ lie in $T$ and contribute two points to $B$, and each class $X$ with $n_X=2$ contributes the point of its line $L_X$. No other $P$-line lies in $T$, since a class with $n_X\le1$ contains no two points of $T$, while in a class with $n_X=2$ the second $P$-line misses $T$.  Hence $|B|=2+\kappa=14-s$.  Likewise, a class $X\neq O$ with $n_X>0$ meets $T$ in a single point ($n_X=1$) or
in the single $P$-line $L_X$ ($n_X=2$), so it contributes exactly one point to $E$, and $\alpha$ contributes two: $|E|=2+s$.  By~\eqref{eq:srange},
\[
|B|+|E|=16,\qquad 6\le|B|\le8,\qquad B\subseteq E .
\]

\smallskip
By Lemma~\ref{lem:BC} there is a plane $\tau\subseteq D(B,E)\subseteq\widehat\Pi$. Write $\tau=\PGv{\widehat U}$ with $\widehat U\le\widehat V$ of dimension $3$, and let $U\le\GF2^6$ be the preimage of $\widehat U$, of dimension $4$ and containing $p$.  We claim that the solid $\PGv U$ lies in $D(T)$.

First, $P=[p]\in\ell\subseteq D(T)$ by~\eqref{eq:ellD}.  Every other point of $\PGv U$ is $[u]$ with $\widehat u\neq0$ and $R=[\widehat u]\in\tau$, and the points of $\PGv U$ lying over $R$ are exactly the pair $[u]$, $[u+p]$.  Since $\tau\subseteq D(B,E)$, there are $\beta\in B$ and $\gamma\in E$ with $\beta\neq\gamma$ and $R=[\beta+\gamma]$.  Choose $\la$ with $\widehat\la=\beta$, so that $\{\la,\la+p\}\subseteq T$, and a point $\mu\in T$ with $\widehat\mu=\gamma$.  The vectors $\la+\mu$ and $\la+\mu+p$ are nonzero (as $\beta\neq\gamma$), differ by $p$, and reduce to $\beta+\gamma$ modulo $\langle p\rangle$. Therefore $\{[\la+\mu],[\la+\mu+p]\}$ is exactly the pair of points of $\PGv U$ over $R$, and both are differences of points of $T$.  As $R$ ranges over the seven points of $\tau$, this accounts for all $14$ points of $\PGv U\setminus\{P\}$, proving the claim and the theorem. 
\end{proof}

The extension theorem now follows exactly as in Section~\ref{sec:93}.

\begin{thm}\label{thm:43}
Let $C$ be a nondegenerate additive $(n,3,d)_{4/2}$-code.  Then $C$ is extendable if and only if $C$ admits an additive extension.
\end{thm}

\begin{proof}
One implication is trivial.  For the other, suppose $C$ is extendable. By Lemma~\ref{lem:transpartition} there is a partition of $\AG62$ into $4$ transversals of $\Afold$. One class contains at least $64/4=16$ points, so we may form a transversal $T$ of $\Afold$ with $|T|=16$.  By Theorem~\ref{thm:dir16}, $D(T)$ contains a solid $\Lambda$ of $\Pi$, and since $D(T)\cap\Afold=\emptyset$ we get $\Lambda\cap\Afold=\emptyset$.  Now Proposition~\ref{prop:addext} (with $km-m-1=3$) provides an additive extension.
\end{proof}

\begin{cor}\label{cor:43max}
A nondegenerate additive $(n,3,d)_{4/2}$-code is maximal if and only if it is additively maximal, if and only if every solid of $\PG52$ meets $\Afold$, if and only if its projective system of lines in $\PG52$ is complete.
\end{cor}

\begin{rem}\label{rem:contrast}
Theorems~\ref{thm:dir16} and~\ref{thm:43} contrast sharply with Section~\ref{sec:counter}, which concerns the \emph{same} ambient geometry. There, $(q,m,k)=(2,3,2)$, so transversals again live in $\AG62$, but the fibres of a new faithful coordinate have $8$ points and the null flats are planes of $\PG52$, and we exhibit an $8$-point set of
$\AG62$ whose direction set contains no plane (the set $T$ of Proposition~\ref{prop:ext} determines exactly the directions of weights $1$, $2$, $5$, and $6$, and by Lemma~\ref{lem:block34} its complement meets every plane).  Thus in $\AG62$ every $16$-point set determines all the points of a solid, while an $8$-point set need not even determine a plane.
\end{rem}

\section{Scattered Linear Sets:\\ Counterexamples for Non-Prime $q$}
\label{sec:scatter}

In this section $q=q_0^{\,e}$ is a proper prime power ($e\geq2$), and we work with the parameters $m=2$, $k=2$, so that $\Pi=\PG{3}{q}$, null flats are lines of $\Pi$, and the fibres of a new faithful coordinate have $q^{km-m}=q^2$ points.  We show that for every square $q$, there exists an extendable additive code with no additive extension (Theorem~\ref{thm:scattercode}).

Recall (see~\cite{MR1772204,MR2684078}) that for a $\GF{q_0}$-subspace $W$ of $\GF{q}^{4}$ the associated \textit{linear set} is
\[
  L_W\;=\;\bigl\{[w] : w\in W\setminus\{0\}\bigr\}\;\subseteq\;\Pi ,
\]
of \textit{rank} $\dim_{\GF{q_0}}W$, and that $W$ (or $L_W$) is \textit{scattered} if $w\mapsto[w]$ is injective up to $\GF{q_0}$-scalars, that is, if $|L_W|=(q_0^{\,\mathrm{rank}\,W}-1)/(q_0-1)$; equivalently, if $\langle v\rangle_{\GF{q}}\cap W$ has $\GF{q_0}$-dimension at most $1$ for every $v\neq0$.  Throughout this section $\sigma:x\mapsto x^{q_0}$ denotes the $q_0$-Frobenius map of $\GF{q}$ and
\begin{equation}\label{eq:defU}
  U \;=\; \bigl\{\,(x,\;y,\;x^{\sigma},\;y^{\sigma}) :  x,y\in\GF{q}\,\bigr\}\;\subseteq\;\GF{q}^{4}.
\end{equation}

\begin{lem}\label{lem:scattered}
$U$ is a $\GF{q_0}$-subspace of $\GF{q}^4$ with $|U|=q^2$; it spans $\GF{q}^4$ over $\GF{q}$, and it is scattered. \\ In particular, since $U-U=U$, the set of directions determined by the $q^2$ affine points of $U\subseteq\AG{4}{q}$ is the linear set $D(U)=L_U$, of size $(q^{2}-1)/(q_0-1)$.
\end{lem}

\begin{proof}
Since $\sigma$ is additive and fixes $\GF{q_0}$ elementwise, $U$ is $\GF{q_0}$-linear, and clearly $|U|=q^2$.  For the spanning claim pick $\omega\in\GF{q}$ with $\omega^{\sigma}\neq\omega$ (possible as $e\geq2$). Then $(1,0,1,0)$ and $(\omega,0,\omega^{\sigma},0)$ are $\GF{q}$-independent, so the $\GF{q}$-span of $U$ contains the coordinate plane $\{(a,0,b,0)\}$, and similarly it contains $\{(0,a,0,b)\}$.  Finally, suppose $cu\in U$ where $u=(x,y,x^{\sigma},y^{\sigma})\neq0$ and $c\in\GF{q}^{*}$, say $x\neq0$.  Comparing the first and third coordinates of $cu$ gives $(cx)^{\sigma}=c\,x^{\sigma}$, whence $c^{\sigma}=c$ and $c\in\GF{q_0}$.  Thus $\langle u\rangle_{\GF q}\cap U=\GF{q_0}u$ for every $u\in U\setminus\{0\}$, and $U$ is scattered.
\end{proof}

\begin{prop}\label{prop:linefree}
If $q=q_0^{\,e}$ with $e\geq2$, then $D(U)$ contains no line of $\Pi=\PG{3}{q}$.  
\end{prop}

\begin{proof}
Let $\ell=\PGv{L}$ be a line of $\Pi$, where $L$ is a $2$-dimensional $\GF{q}$-subspace of $\GF{q}^4$.  If $[u]\in\ell\cap L_U$ with $u\in U\setminus\{0\}$, then $u\in L$; hence $\ell\cap L_U=L_{U\cap L}$, and if $j=\dim_{\GF{q_0}}(U\cap L)$ then $|\ell\cap L_U|=(q_0^{\,j}-1)/(q_0-1)$ by Lemma~\ref{lem:scattered}. If $\ell\subseteq L_U$ then
\[
  \frac{q_0^{\,j}-1}{q_0-1}\;=\;q+1\;=\;q_0^{\,e}+1, \qquad\text{i.e.}\qquad q_0^{\,j}\;=\;q_0\bigl(q_0^{\,e}-q_0^{\,e-1}+1\bigr).
\]
For $e\geq2$, $1<q_0^{\,e}-q_0^{\,e-1}+1 \equiv 1 \pmod q_0$ is not a power of $q_0$.  This contradiction proves the claim.
\end{proof}


\begin{rem}
For $e=1$ the displayed equation has the solution $j=2$. A $2$-dimensional $\GF{q}$-subspace determines exactly the points of one line of $\Pi$, in accordance with Lemma~\ref{lem:dir9}.  The set~\eqref{eq:defU} is the classical example of a \emph{maximum scattered} linear set (of pseudoregulus type), see~\cite{MR1772204,LMPT14,MR2684078}.
\end{rem}

\subsection*{The code, for square $q$}

For the remainder of the section let $e=2$ and write $q=t^{2}$ (so $t=q_0$), so that $U$ is scattered of rank $4$ and $|L_U|=t^{3}+t^{2}+t+1$.  Call a line of $\Pi$ \textit{external}, \textit{tangent} or \textit{secant} to $L_U$ according as it meets $L_U$ in $0$, $1$ or at least $2$ points.

\begin{lem}\label{lem:secants}
If $q=t^2$, then:
\begin{itemize}
\item[(i)] every secant meets $L_U$ in exactly $t+1$ points, and $W\mapsto\langle L_W\rangle$ is a bijection from the rank-$2$ $\GF{t}$-subspaces of $U$ onto the secants. In particular the points of $L_U$ and the secants form a projective space $\PG{3}{t}$, and there are exactly $(t^{2}+1)(t^{2}+t+1)$ secants; \item[(ii)] every point of $\Pi\setminus L_U$ lies on exactly one secant;
\item[(iii)] every point of $\Pi\setminus L_U$ lies on exactly $t^{3}(t-1)$ external lines, and the total number of external lines is
\[
  n_t \;=\; t^{4}(t-1)^{2}(t^{2}+t+1).
\]
\end{itemize}
\end{lem}

\begin{proof}
(i)  For a line $\ell=\PGv{L}$ we saw that $\ell\cap L_U=L_{U\cap L}$ has $(t^{\,j}-1)/(t-1)$ points, where $j=\dim_{\GF{t}}(U\cap L)$.  Since $(t^{3}-1)/(t-1)=t^2+t+1$ exceeds the number $t^{2}+1$ of points of $\ell$, necessarily $j\leq2$, so $|\ell\cap L_U|\in\{0,1,t+1\}$.  If $[u]\neq[v]$ are points of $\ell\cap L_U$, then $W=\langle u,v\rangle_{\GF t}\leq U\cap L$ has rank $2$, the $t+1$ points of $L_W$ lie on $\ell$, and $U\cap L=W$ by $j\le2$, so $\ell\cap L_U=L_W$ and $\ell=\langle L_W\rangle$. Conversely, for any rank-$2$ subspace $W\leq U$ the $t+1$ points of $L_W$ are collinear (as $[au+bv]$ lies on the line through $[u]$ and $[v]$), so $\langle L_W\rangle$ is a secant, and distinct $W$ give distinct secants, since $\langle L_{W_1}\rangle=\langle L_{W_2}\rangle=\PGv{L}$ forces $W_1=U\cap L=W_2$.  By scatteredness the natural map $\PGv{U}\to L_U$ is a bijection, so points of $L_U$ with the secants provide $\PGv{U}\cong\PG3t$, which contains $(t^2+1)(t^2+t+1)$ lines.

(ii)  Suppose first that distinct secants $g_1,g_2$ pass through a common point $P\in\Pi\setminus L_U$, and write $g_i\cap L_U=L_{W_i}$ as in (i).  If $W_1\cap W_2\neq\{0\}$, pick  $0\ne w\in W_1\cap W_2$. Then $[w]\in g_1\cap g_2=\{P\}$, so $P\in L_U$, a contradiction.  If $W_1\cap W_2=\{0\}$, then $W_1+W_2$ has $\GF{t}$-dimension $4$, so $U=W_1+W_2$. On the other hand the concurrent lines $g_1\neq g_2$ span a plane $\PGv{H}$ with $\dim_{\GF q}H=3$, and $W_1+W_2\subseteq H$, so $U\subseteq H$, contradicting the spanning claim of Lemma~\ref{lem:scattered}.  Hence each point of $\Pi\setminus L_U$ lies on at most one secant.  Now count incidences: each secant contains $t^2+1-(t+1)=t(t-1)$ points of $\Pi\setminus L_U$, and
\[
  |\Pi\setminus L_U| =(t^4+1)(t^2+1)-(t^2+1)(t+1) =t(t^2+1)(t^3-1),
\]
so the average number of secants through a point of $\Pi\setminus L_U$ is
\[
  \frac{(t^2+1)(t^2+t+1)\cdot t(t-1)}{t(t^2+1)(t^3-1)}\;=\;1 .
\]
Combined with the argument above, every such point lies on exactly one secant.

(iii)  Fix $P\in\Pi\setminus L_U$.  Of the $q^2+q+1=t^4+t^2+1$ lines of $\Pi$ through $P$, exactly one is a secant, and it absorbs $t+1$ points of $L_U$.  Each of the remaining $|L_U|-(t+1)=t^3+t^2$ points of $L_U$ lies on precisely one line through $P$, and such a line contains no second point of $L_U$ (it would otherwise be a second secant through $P$). Hence there are exactly $t^3+t^2$ tangents through $P$, and
\[
  (t^4+t^2+1)-1-(t^3+t^2)\;=\;t^4-t^3\;=\;t^3(t-1)
\]
external lines through $P$.  Finally, every point of an external line lies in $\Pi\setminus L_U$, so counting incident pairs (point of $\Pi\setminus L_U$, external line) gives \[n_t\,(t^2+1)=t(t^2+1)(t^3-1)\cdot t^3(t-1),\] i.e.\ $n_t=t^4(t-1)(t^3-1)=t^4(t-1)^2(t^2+t+1)$.
\end{proof}

\begin{defn}\label{defn:codeCt}
Let $\ell_1,\ldots,\ell_{n_t}$ be the external lines of $L_U$, each taken once.  For each $j$ choose a surjective $\GF{q}$-linear map $\rho_j:\GF{q}^4\to\GF{q}^2\cong\GF{q^2}$ whose kernel $V_j$ satisfies $\PGv{V_j}=\ell_j$, and let
\[
  C_t \;=\;\bigl\{\,\bigl(\rho_1(\la),\ldots,\rho_{n_t}(\la)\bigr) :
  \la\in\GF{q}^{4}\,\bigr\}\;\subseteq\;\GF{q^2}^{\,n_t}.
\]
\end{defn}

\begin{thm}\label{thm:scattercode}
Let $q=t^2$ with $t\geq2$ a prime power, and put $n=n_t=t^4(t-1)^2(t^2+t+1)$ and $d=n-t^3(t-1)$.  Then $C_t$ is a nondegenerate additive $(n,2,d)_{q^2/q}$-code whose dual system consists of the external lines of $L_U$ and for which
\[
  \Afold\;=\;\Pi\setminus L_U .
\]
Any two distinct codewords of $C_t$ are at distance $d$ or $n$.\\ 
The code $C_t$ is extendable (the $q^2$ cosets of $U$ partition $\AG4q$ into transversals of $\Afold$, yielding an $(n+1,2,d+1)_{q^2}$-extension), but $C_t$ admits no additive extension.\\
In particular, for $t=2$ there is an extendable additive $(112,2,104)_{16/4}$-code with no additive extension, and for $t=3$ an extendable additive $(4212,2,4158)_{81/9}$-code with no additive extension.
\end{thm}

\begin{proof}
By Lemma~\ref{lem:agree}, the codewords indexed by $\la\neq\mu$ agree in exactly the number of null lines $\ell_j$ through the point $[\la-\mu]$. Since the $\ell_j$ are precisely the external lines, this number is $0$ if $[\la-\mu]\in L_U$ and $t^3(t-1)$ otherwise (Lemma~\ref{lem:secants}(iii)).  As $t^3(t-1)<n$, distinct messages give distinct codewords, so $|C_t|=q^4$, and the distances between distinct codewords are $n$ and $n-t^3(t-1)=d$, both attained.  Hence $C_t$ is a nondegenerate additive $(n,2,d)_{q^2/q}$-code (each $\rho_j$ being surjective), $n-d=t^3(t-1)$, and the $(n-d)$-fold points of the dual system are exactly the points of $\Pi\setminus L_U$.

By Proposition~\ref{prop:addext}, an additive extension of $C_t$ requires a line of $\Pi$ disjoint from $\Afold=\Pi\setminus L_U$, that is, a line contained in $L_U$; no such line exists, by Proposition~\ref{prop:linefree} (or directly by Lemma~\ref{lem:secants}(i)).  Hence $C_t$ admits no additive extension.

Finally, the $q^2$ cosets of $U$ partition $\GF{q}^4=\AG4q$ into $q^m$ classes of size $q^2$.  Two codewords in a common class have difference in $U\setminus\{0\}$, hence direction in $L_U$, and so agree in no coordinate, and are at distance $n\geq d+1$.  By Lemma~\ref{lem:partition}, $C_t$ extends to an $(n+1,2,d+1)_{q^2}$-code. 
\end{proof}

\begin{rem}\label{rem:subfield}
The extension just constructed appends to the codeword of $\la$ the coset $\la+U$, and the quotient map $\GF{q}^4\to\GF{q}^4/U$ is $\GF{t}$-linear onto a group of order $q^2$.  Identifying $\GF{q}^4/U$ with $\GF{q^2}$ as $\GF{t}$-vector spaces, the extended code is therefore additive over $\GF{t}$.  Since $C_t$ is $\GF{q}$-linear, it is in particular an additive $(n,2,d)_{t^4/t}$-code, and \emph{as such} it admits an additive extension, while as an $(n,2,d)_{t^4/t^2}$-code it does not.  Additive extendability is thus sensitive to the declared field of linearity, and can be lost in passing from a subfield to a larger one.
\end{rem}

\begin{rem}\label{rem:properlyadd}
The codes $C_t$ are properly additive. Indeed, if $C_t$ were monomially (or semilinearly) equivalent to a $\GF{q^2}$-linear code $L$, then $L$ would be extendable, hence linearly extendable by the theorem of Alderson and G\'acs~\cite{MR2529622}. A linear extension is in particular an additive extension, and pulling it back through the equivalence (which preserves additivity) would give an additive extension of $C_t$.  The same argument shows that the code of Section~\ref{sec:counter} is properly additive.
\end{rem}

\section{An Extendable Additive Code with No Additive Extension}
\label{sec:counter}

In this section we take $q=2$, $m=3$, $k=2$, so that additive codes are $\GF2$-linear subspaces of $\GF8^n$ with $2^6$ codewords, $\Sigma=\PG62$, $\Pi=\PG52$, null flats are planes ($2$-flats) of $\Pi$, and the fibres of a new faithful coordinate have $q^{km-m}=8$ points.  We construct a nondegenerate additive $(30,2,24)_{8/2}$-code that is extendable but admits no additive extension.

Throughout, $e_1,\ldots,e_6$ is the standard basis of $\GF2^6$, $\onevec=e_1+\cdots+e_6$ is the all-one vector, $\wt$ denotes Hamming weight, and we identify a nonzero vector $v$ with the point $[v]\in\Pi$ and a $3$-dimensional subspace $W\le\GF2^6$ with the plane $\PGv{W}\subseteq\Pi$.


We begin with two tactical lemmas.

\begin{lem}\label{lem:block34}
Every $3$-dimensional subspace of $\GF2^6$ contains a nonzero vector of weight $3$ or $4$.
\end{lem}

\begin{proof}
Suppose $W$ is a $3$-dimensional subspace all of whose nonzero vectors have weight in $\{1,2,5,6\}$.  The even-weight vectors of $W$ form a subspace $W_e$ of index at most $2$.

\smallskip
\noindent\textit{Case $W=W_e$.}  All seven nonzero vectors have weight $2$ or $6$.  At most one vector has weight $6$ (namely $\onevec$), so at least six vectors have weight $2$. Regard these as edges of a graph on the vertex set $\{1,\ldots,6\}$ of coordinate positions.  The sum of two distinct weight-$2$ vectors has weight $4$ if the edges are disjoint and weight $2$ if they share a vertex. As $W$ contains no weight-$4$ vector,
the (at least six) edges are pairwise intersecting.  A family of more than three pairwise intersecting edges is a star, so these edges pass through a common vertex $x$. But then for distinct edges $\{x,a\},\{x,b\}$ in $W$ the sum is the edge $\{a,b\}\in W$, which is disjoint from a third star edge $\{x,c\}$---so their sum has weight $4$, a contradiction.

\smallskip
\noindent\textit{Case $[W:W_e]=2$.}  Here, $W$ has four vectors of odd weight, each of weight $1$ or $5$, i.e.\ of the form $e_a$ or $\bar e_a:=\onevec+e_a$.  For $a\ne b$ the sum $e_a+\bar e_b= \onevec+e_a+e_b$ has weight $4$, therefore if vectors of both types occur in $W$, then they occur only as a complementary pair $e_a,\bar e_a$, and a third odd vector of either type is impossible.  All four odd vectors are therefore of the same type.  Being a coset of $W_e$, the four odd vectors sum to $0$. However, for distinct $a,b,c,d$, $e_a+e_b+e_c+e_d$ has weight $4$, and $\bar e_a+\bar e_b+\bar e_c+\bar e_d=e_a+e_b+e_c+e_d$, neither is $0$.  This contradiction completes the proof.
\end{proof}

Let
\[
  \W \;=\; \bigl\{\,W\le\GF2^6 : \dim W=3,\ \wt(v)\in\{3,4\}\ \text{for all } v\in W\setminus\{0\}\,\bigr\}.
\]

\begin{lem}\label{lem:K4}
Identify the six coordinate positions with the edges of the complete graph $K_4$ on vertices $\{1,2,3,4\}$, and for a vertex $u$ let $s_u\in\GF2^6$ be the characteristic vector of the star of $u$ (the three edges at $u$).  Then
\[
  W \;=\; \{0\}\cup\{s_u : u\}\cup\{s_u+s_v : u\neq v\}
\]
is a member of $\W$, every member of $\W$ arises from exactly one identification up to automorphisms of $K_4$, and
\[
  |\W| \;=\; \frac{6!}{|\mathrm{Aut}(K_4)|}\;=\;\frac{720}{24}\;=\;30 .
\]
Moreover each member of $\W$ contains exactly four weight-$3$ and three weight-$4$ vectors, and each weight-$3$ (respectively weight-$4$) vector of $\GF2^6$ lies in exactly $6$ members of $\W$.
\end{lem}

\begin{proof}
Each star $s_u$ has weight $3$.  For $u\neq v$, the vector $s_u+s_v$ is the characteristic vector of the symmetric difference of two stars, namely the four edges meeting $\{u,v\}$ in exactly one vertex, and it has weight $4$.  For distinct $u,v,w$ with fourth vertex $x$, each edge within $\{u,v,w\}$ is counted twice in $s_u+s_v+s_w$ and each edge at $x$ once, so $s_u+s_v+s_w=s_x$.  Hence $\mathrm{span}(s_u,s_v,s_w)$ has the listed seven nonzero vectors, is $3$-dimensional, and lies in $\W$.

Conversely, let $W\in\W$.  We claim $W$ contains exactly four vectors of weight $3$.  If all seven nonzero vectors had weight $4$, then each of the six coordinate functionals, being linear on $W$, would be either zero or equal to $1$ on exactly $4$ of the $8$ vectors of $W$. Summing weights gives $7\cdot4=28=4r$ with $r$ the number of nonvanishing coordinate functionals, giving $r=7>6$.  Hence the even-weight subspace $W_e$
has index $2$, $W$ has four (odd) vectors of weight $3$ and three nonzero (even) vectors of weight $4$.

Let $A_1,A_2,A_3,A_4\subseteq\{1,\ldots,6\}$ be the supports of the four weight-$3$ vectors.  For $i\neq j$ the sum of the corresponding vectors lies in $W$ and is even, of weight $6-2|A_i\cap A_j|=4$, giving $|A_i\cap A_j|=1$.  The four odd vectors form a coset of $W_e$, so they sum to $0$, so every position lies in an even number of the $A_i$.  No position $t$ lies in all four (otherwise $A_i\cap A_j=\{t\}$ for all $i\ne j$ and the sets $A_i\setminus\{t\}$ are pairwise disjoint of size $2$, requiring $1+8>6$ positions).  Counting incidences, $\sum_i|A_i|=12$, so each of the six positions lies in exactly two of the $A_i$.  Form the graph with vertices $A_1,\ldots,A_4$ and one edge per position, joining the two sets containing it. There are $4$ vertices, $6$ edges, and any two vertices are joined by exactly $|A_i\cap A_j|=1$ edge. Thus it is $K_4$, the positions are its edges, and $A_i$ is the star of the vertex $i$.  As such, $W$ arises from an identification of the positions with $E(K_4)$, and the identification is unique up to $\mathrm{Aut}(K_4)$ as $W$ determines its weight-$3$ vectors, i.e.\ the star structure. Since distinct star-quadruples give distinct subspaces, $|\W|=6!/|\mathrm{Aut}(K_4)|=30$.

Finally, the symmetric group $S_6$ permutes the coordinate positions, preserves $\W$, and acts transitively on the weight-$3$ vectors of $\GF2^6$ as well as on the weight-$4$ vectors.  Hence the number of members of $\W$ through a fixed weight-$3$ vector is a constant $c_3$, and double counting gives $30\cdot4=20\,c_3$, so $c_3=6$. Similarly $30\cdot3=15\,c_4$ gives $c_4=6$.
\end{proof}

\subsection*{The code}

\begin{defn}\label{defn:codeC}
Write $\W=\{W_1,\ldots,W_{30}\}$.  For each $j$ choose a surjective $\GF2$-linear map $\rho_j:\GF2^6\to\GF2^3\cong\GF8$ with $\ker\rho_j=W_j$, and let
\[
  C \;=\;\bigl\{\,\bigl(\rho_1(\la),\ldots,\rho_{30}(\la)\bigr) :
  \la\in\GF2^6\,\bigr\}\;\subseteq\;\GF8^{30}.
\]
\end{defn}

\begin{prop}\label{prop:codeC}
$C$ is a nondegenerate additive $(30,2,24)_{8/2}$-code whose dual system consists of the $30$ planes $\PGv{W_j}$, and
\[
  \Afold \;=\; \bigl\{[v] : \wt(v)\in\{3,4\}\bigr\},
\]
a set of $35$ points of $\Pi=\PG52$.  Any two distinct codewords of $C$ are at distance $24$ or $30$.
\end{prop}

\begin{proof}
Two codewords indexed by $\la\neq\mu$ agree in exactly $\#\{j : \la-\mu\in W_j\}$ coordinates by Lemma~\ref{lem:agree}, and
\[
  \#\{j : \la-\mu\in W_j\} \;=\;
  \begin{cases}
    6, & \wt(\la-\mu)\in\{3,4\},\\
    0, & \text{otherwise},
  \end{cases}
\]
 by Lemma~\ref{lem:K4} (and by the definition of $\W$ no vector of weight outside $\{3,4\}$ lies in any $W_j$).  In particular distinct message vectors give distinct codewords, $|C|=2^6$, and the distances between distinct codewords are $30-6=24$ and $30-0=30$, both attained.  Hence $d=24$, $n-d=6$, and the $(n-d)$-fold points are exactly the points $[v]$ with $\wt(v)\in\{3,4\}$.
\end{proof}

\begin{prop}\label{prop:noadd}
$C$ admits no additive extension.
\end{prop}

\begin{proof}
By Proposition~\ref{prop:addext}, an additive extension requires a plane of $\Pi$ disjoint from $\Afold$, i.e.\ a $3$-dimensional subspace of $\GF2^6$ containing no vector of weight $3$ or $4$.  No such subspace exists by Lemma~\ref{lem:block34}.
\end{proof}

\begin{prop}\label{prop:ext}
$C$ is extendable.  Explicitly, let
\[
  T \;=\; \{0,\ e_1,\ e_2,\ \ldots,\ e_6,\ \onevec\}
  \qquad\text{and}\qquad
  V \;=\;\langle 110100,\ 011010,\ 001101\rangle .
\]
Then $V\in\W$, the translates $\{T+v : v\in V\}$ partition $\AG62$ into eight transversals of $\Afold$, and the corresponding extension of $C$ is a $(31,2,25)_8$-code.
\end{prop}

\begin{proof}
The differences of distinct elements of $T$ are the vectors $e_a$ (weight~$1$), $e_a+e_b$ (weight~$2$), $\onevec+e_a$ (weight~$5$) and $\onevec$ (weight~$6$). Hence $D(T)\cap\Afold=\emptyset$, and the same holds for every translate $T+v$, since $D(T+v)=D(T)$.

The listed generators of $V$ have weight $3$, their pairwise sums $101110$, $111001$, $010111$ have weight $4$, and their total sum $100011$ has weight $3$, so $V\in\W$.  (In the notation of Lemma~\ref{lem:K4}, $V$ is the member of $\W$ associated with a suitable labelling of $E(K_4)$, any member of $\W$ would serve.)  Since every nonzero vector of $V$ has weight $3$ or $4$ and every difference of points of $T$ has weight in $\{0,1,2,5,6\}$, we get $(T-T)\cap V=\{0\}$. Therefore, the map $T\times V\to\GF2^6$, $(t,v)\mapsto t+v$, is injective, and by cardinality the translates $T+v$, $v\in V$, partition $\AG62$.

Thus $\AG62$ is partitioned into $q^m=8$ transversals of $\Afold$, and $C$ is extendable by Lemma~\ref{lem:transpartition}.  Concretely, the extension appends to the codeword of $\la$ the unique $v\in V$ with $\la\in T+v$ (an alphabet of size $8$). Two words in a common fibre have difference in $(T-T)\setminus\{0\}$, hence fold number $0$ and distance $30\geq25$. Two words in distinct fibres are at distance at least $24+1=25$, and a pair at distance $24$ in $C$ (e.g.\ $\la-\mu$ of weight~$3$) lies in distinct fibres, so the distance $25$ is attained.
\end{proof}

Combining Propositions~\ref{prop:codeC}--\ref{prop:ext}:

\begin{thm}\label{thm:counter}
The code $C$ of Definition~\ref{defn:codeC} is a nondegenerate additive $(30,2,24)_{8/2}$-code that is extendable but admits no additive extension.  In particular:
\begin{itemize}
\item[(i)] the Alderson--G\'acs theorem (``extendable $\To$ linearly extendable'') does not extend to additive codes in general; \item[(ii)] an additively maximal additive code need not be maximal, and a complete projective system of $(m-1)$-flats need not correspond to a maximal code.
\end{itemize}
\end{thm}

\begin{rem}
The code $C$ is a two-distance code with distances $24$ and $30$, and the fibres of its extension are the translates of $T=B_1(0)\cup\{\onevec\}$, the Hamming ball of radius one together with its antipode.  It would be interesting to know whether $30$ is the smallest length of an extendable, additively maximal additive code, and more generally for which parameters $(q,m,k)$ such codes exist.
\end{rem}

\section{The Prime Case}
\label{sec:prime}

For $m=2$ the construction of Section~\ref{sec:scatter} rests on the existence of a suitable scattered linear set, which exists precisely because $\GF{q}$ there has a proper subfield, while the example of Section~\ref{sec:counter} has $m=3$.  For $m=2$ over a prime field neither method is available, and we conjecture:

\begin{conj}\label{conj:m2}
Let $p$ be a prime.
\begin{itemize}
\item[(i)] Every set of $p^{2}$ points of $\AG{4}{p}$ determines all points of some line of $\PG{3}{p}$.
\item[(ii)] Consequently, every extendable additive $(n,2,d)_{p^2/p}$-code admits an additive extension.
\end{itemize}
\end{conj}

Part (i) holds for $p=2,3$ by Lemma~\ref{lem:dir9} and Proposition \ref{prop:42max} (Section~\ref{sec:93}).  A naive nonexhaustive computer search found no counterexample for $p=5$.  A natural approach to  Conjecture~\ref{conj:m2} is through the theory of directions and R\'edei type blocking sets.  For a set $T$ of $p^2$ points of $\AG4p$, $D(T)$ contains no line of $\PG{3}{p}$ if and only if the set $B=\Pi\setminus D(T)$ of undetermined directions meets every line of $\PG{3}{p}$.  Projecting $T$ from an undetermined direction $[u]\in B$ yields a set of $p^{2}$ points of $\AG{3}{p}$, a set of R\'edei size, where the structure theorem of Storme and Sziklai~\cite{MR1869411} (the directions determined by a set of $q^{2}$ points of $\AG{3}{q}$ form a union of full lines) and the prime-field direction theorems~\cite{MR1682973,MR2019280} apply; see~\cite{SzT12} for a related higher-dimensional direction problem.

\begin{rem}\label{rem:graphcase}
Conjecture~\ref{conj:m2}(i) splits into two cases that perhaps provide insight into why primality may be key.  Let $T$ be a putative counterexample and let $B=\Pi\setminus D(T)$ be its set of undetermined directions, so that $B$ meets every line of $\Pi$.

Suppose first that $B$ contains a line $\ell$.  Then no difference of points of $T$ lies in $\ell$, so the projection of $\AG{4}{p}$ along $\ell$ is injective on $T$, hence bijective onto $\AG{2}{p}$. In suitable coordinates $\GF{p}^{4}=\GF{p}^{2}\times\GF{p}^{2}$ we have $T=\{(x,f(x)):x\in\GF{p}^{2}\}$ for an arbitrary function $f:\GF{p}^{2}\to\GF{p}^{2}$, with $\ell=\PGv{0\times\GF{p}^{2}}$.  A line of $\Pi$ contained in $D(T)$ must avoid $B\supseteq\ell$, and the lines of $\Pi$ disjoint from $\ell$ are precisely the sets $\PGv{\mathrm{graph}(A)}$, where $A:\GF{p}^{2}\to\GF{p}^{2}$ is linear and $\mathrm{graph}(A)=\{(v,Av):v\in\GF{p}^{2}\}\le\GF{p}^{4}$ (a $2$-dimensional subspace meets $0\times\GF{p}^{2}$ trivially exactly when it is such a graph). Moreover $[(v,Av)]\in D(T)$ if and only if $f-A$ identifies two points of some affine line of direction $v$. This case of Conjecture~\ref{conj:m2}(i) is therefore the statement: \emph{for every $f:\GF{p}^{2}\to\GF{p}^{2}$ there is a linear map $A$ such that $f-A$ is non-injective on some line of every parallel class of $\AG{2}{p}$.}  This is a direction problem for vector-valued functions, and the prime-field slope theorems bear on it directly. For every affine line $\ell'$ of the domain and every linear functional $\la$ on the codomain, the graph of $\la\circ(f-A)\vert_{\ell'}$ is a set of $p$ points of $\AG{2}{p}$, which by R\'edei--Megyesi~\cite{MR0352060} and its refinements~\cite{MR1682973,MR2019280} is affine or determines at least $(p+3)/2$ slopes; see~\cite{MR2375464} for functions of several variables.  The subfield-linear maps occupying the exceptional middle range of the slope theorem of~\cite{MR2019280} are precisely the source of the counterexamples of Section~\ref{sec:scatter}, where $T=U$ is the graph of the $\GF{q_0}$-linear map $(x,y)\mapsto(x^{q_0},y^{q_0})$.  Over prime fields that range is empty.

If instead $B$ contains no line, then $B$ is a blocking set of $\PG{3}{p}$ with respect to lines containing no full line, so the equality case of the Bose--Burton theorem (a plane) is excluded and $|B|\geq p^{2}+p+2$. Moreover, projecting $T$ from any $[u]\in B$ yields $p^{2}$ points of $\AG{3}{p}$ determining \emph{all} of their directions, since an undetermined direction of the projected set corresponds to a line of $\Pi$ through $[u]$ contained in $B$.
\end{rem}

Primality is not the only hypothesis that removes the scattered obstruction of Section~\ref{sec:scatter}. Indeed,  raising the dimension $k$ removes it as well, for \emph{every} $q$.  We therefore close this section with the direction problem for $m=2$ and $k\geq3$, concerning sets of $q^{2k-2}$ points of $\AG{2k}{q}$ and $(2k-3)$-flats of $\PG{2k-1}{q}$. The first instance, $(q,k)=(2,3)$, is settled affirmatively by Theorem~\ref{thm:dir16}.  Indeed, for non-prime $q=q_0^{\,e}$ a transversal has $q^{2k-2}=q_0^{(2k-2)e}$ points, and $(2k-2)e$ exceeds the maximum rank $ke$ of a scattered $\GF{q_0}$-linear set in $\GF{q_0^e}^{2k}$ when $k\geq3$ (see~\cite{MR1772204}).  Thus no analogue of the counterexamples of Section~\ref{sec:scatter} can arise from linear sets once $k\geq3$, and the case $k\geq3$ rests on the same footing as the prime case, where the known obstructions are absent.

\section{Concluding Remarks and Open Problems}
\label{sec:remarks}

For $m=1$  (linear codes) extendability always implies additive (linear) extendability~\cite{MR2529622}.  For $m=2$ the same holds for $q\in\{2,3\}$ when $k=2$ (Section~\ref{sec:93}) and for $(q,k)=(2,3)$ (Section~\ref{sec:43}), while for every square $q$ it does not (Section~\ref{sec:scatter}); we conjecture that it holds for all primes (Conjecture~\ref{conj:m2}).  For $m\geq3$ it does not hold even over the prime field: $(q,m,k)=(2,3,2)$ (Section~\ref{sec:counter}).  We collect the main open questions.

\begin{prob}
Prove or disprove Conjecture~\ref{conj:m2}.  More generally, settle the case $m=2$, $k\geq3$: the first instance $(q,k)=(2,3)$ is Theorem~\ref{thm:43}, and the scattered obstruction is unavailable for $k\geq3$ (see the closing remark of Section~\ref{sec:prime}); the cases $q>2$ and $k\geq4$ remain open.
\end{prob}

\begin{prob}\label{prob:nonsquare}
Extend Theorem~\ref{thm:scattercode} to non-square, non-prime $q$.  
Further, determine the minimum length of an extendable, additively maximal additive $(n,2,d)_{q^2/q}$-code for square $q$: is $n=112$ optimal for $q=4$? (Any sub-multiset of the external lines whose set of maximal-fold points still meets every line of $\Pi$ yields a shorter example.)
\end{prob}

\begin{prob}
For which triples $(q,m,k)$ with $m\geq3$ do extendable, additively maximal additive $(n,k,d)_{q^m/q}$-codes exist?  Does the construction of Section~\ref{sec:counter} generalize to all $q$ (with $m=3$, $k=2$), or to $m\geq4$?
\end{prob}

\begin{prob}
Extendability of a code is invariant under equivalence (Remark~\ref{rem:equiv}). Whether \emph{additive} extendability is likewise invariant leads to a rigidity question.  The dual system of a nondegenerate additive code is well defined up to a collineation of $\Pi$ (Remark~\ref{rem:welldef}).  Is it moreover an invariant of the equivalence class? If two nondegenerate additive $(n,k,d)_{q^m/q}$-codes are equivalent in the broad, isometric sense of Remark~\ref{rem:equiv}, then must some collineation of $\Pi$ carry the dual system of one onto that of the other?  Since additive extendability depends only on the dual system (Proposition~\ref{prop:addext}), a positive answer would show that additive extendability, like extendability, is invariant under general code equivalence.
\end{prob}

\begin{prob}\label{prob:rigidity}
For linear codes, more is true than the Alderson--G\'acs theorem: for fixed $k$ and $n-d$, linear codes of sufficient length admit \emph{only} linear extensions (see~\cite{AB1} for the MDS case, \cite{MR2398834,MR2389969} for the AMDS case, and~\cite{AB3} in general).  Is there an additive analogue?
With $q$, $m$, $k$ and $n-d$ all fixed, must every extension of a sufficiently long extendable additive $(n,k,d)_{q^m/q}$-code be additive?  The codes of Sections~\ref{sec:scatter} and~\ref{sec:counter} bound any such length threshold from below.
\end{prob}

Finally, we situate the counterexamples within the geometry of Section~\ref{sec:ext}.  By Corollary~\ref{cor:addmax} a code is additively maximal precisely when $\Afold$ meets every $(km-m-1)$-flat of $\Pi$, while by Proposition~\ref{prop:mflat}, if $\Afold$ contained an $m$-flat, then the code would admit no extension whatsoever.  The codes of Theorems~\ref{thm:scattercode} and~\ref{thm:counter} lie strictly between the two, in that their sets $\Afold$ block every $(km-m-1)$-flat and contain no $m$-flat (as their extendability requires).  
The additive maximality of these codes thus stems from a highly selective blocking property that blocks only the additive extensions.

\bigskip
\noindent\textbf{Acknowledgements.}
 The author acknowledges the support of the Natural Sciences and Engineering Research Council of Canada (NSERC), [funding reference number 2019-04103]\\
Cette recherche a \'{e}t\'{e} financ\'{e}e par le Conseil de recherches en sciences naturelles et en g\'{e}nie du Canada (CRSNG), [num\'{e}ro de r\'{e}f\'{e}rence 2019-04103]

\bibliography{references}

@phdthesis{ALD,
  author        = {Alderson, T. L.},
  title         = {On {MDS} codes and {B}ruen-{S}ilverman codes},
  school        = {University of Western Ontario},
  year          = {2002}
}

@article{AB1,
  author        = {Alderson, T.L. and Bruen, A. A. and Silverman, R.},
  title         = {Maximum distance separable codes and arcs in projective spaces},
  journal       = {J. Combin. Theory Ser. A},
  fjournal      = {Journal of Combinatorial Theory. Series A},
  volume        = {114},
  year          = {2007},
  number        = {6},
  pages         = {1101--1117},
  issn          = {0097-3165},
  mrclass       = {94B27 (51E21)},
  mrnumber      = {MR2337238}
}

@article{AB3,
  author        = {Alderson, T. L. and Bruen, A. A.},
  title         = {Coprimitive sets and inextendable codes},
  journal       = {Des. Codes Cryptogr.},
  fjournal      = {Designs, Codes and Cryptography},
  volume        = {47},
  year          = {2008},
  number        = {1-3},
  pages         = {113--124}
}

@article {MR2529622,
    AUTHOR = {Alderson, T. L. and G{\'a}cs, Andr{\'a}s},
     TITLE = {On the maximality of linear codes},
   JOURNAL = {Des. Codes Cryptogr.},
  FJOURNAL = {Designs, Codes and Cryptography. An International Journal},
    VOLUME = {53},
      YEAR = {2009},
    NUMBER = {1},
     PAGES = {59--68},
      ISSN = {0925-1022},
   MRCLASS = {94B05},
  MRNUMBER = {MR2529622},
       DOI = {10.1007/s10623-009-9293-z},
}

@article {MR2389969,
    AUTHOR = {Alderson, T. L. and Bruen, A. A.},
     TITLE = {Maximal {AMDS} codes},
   JOURNAL = {Appl. Algebra Engrg. Comm. Comput.},
  FJOURNAL = {Applicable Algebra in Engineering, Communication and Computing},
    VOLUME = {19},
      YEAR = {2008},
    NUMBER = {2},
     PAGES = {87--98},
      ISSN = {0938-1279},
   MRCLASS = {94B27},
  MRNUMBER = {MR2389969},
}

@article {MR2398834,
    AUTHOR = {Alderson, T. L. and Bruen, A. A.},
     TITLE = {Codes from cubic curves and their extensions},
   JOURNAL = {Electron. J. Combin.},
  FJOURNAL = {Electronic Journal of Combinatorics},
    VOLUME = {15},
      YEAR = {2008},
    NUMBER = {1},
     PAGES = {Research paper 42, 9},
      ISSN = {1077-8926},
   MRCLASS = {94B27},
  MRNUMBER = {MR2398834},
}

@book{MR0465510,
  author        = {MacWilliams, F. J. and Sloane, N. J. A.},
  title         = {The theory of error-correcting codes},
  publisher     = {North-Holland Publishing Co.},
  address       = {Amsterdam},
  year          = {1977},
  mrclass       = {94A10},
  mrnumber      = {MR0465510}
}

@book{MR2079734,
  author        = {Bierbrauer, Juergen},
  title         = {Introduction to coding theory},
  series        = {Discrete Mathematics and its Applications (Boca Raton)},
  publisher     = {Chapman \& Hall/CRC, Boca Raton, FL},
  year          = {2005},
  pages         = {xxiv+390},
  isbn          = {1-58488-421-5},
  mrclass       = {94-01},
  mrnumber      = {MR2079734}
}

@book{MR2131191,
  author        = {Bruen, A. A. and Forcinito, Mario A.},
  title         = {Cryptography, information theory, and error-correction},
  series        = {Wiley-Interscience Series in Discrete Mathematics and
                  Optimization},
  publisher     = {Wiley-Interscience},
  address       = {Hoboken, NJ},
  year          = {2005},
  pages         = {xxiv+468},
  isbn          = {0-471-65317-9},
  mrclass       = {94-01},
  mrnumber      = {MR2131191}
}

@book{MR1664228,
  author        = {van Lint, J. H.},
  title         = {Introduction to coding theory},
  series        = {Graduate Texts in Mathematics},
  volume        = {86},
  edition       = {Third},
  publisher     = {Springer-Verlag},
  address       = {Berlin},
  year          = {1999},
  pages         = {xiv+227},
  isbn          = {3-540-64133-5},
  mrclass       = {94-01},
  mrnumber      = {MR1664228}
}

@book{1137784,
  author        = {Roth, Ron},
  title         = {Introduction to Coding Theory},
  year          = {2006},
  isbn          = {0521845041},
  publisher     = {Cambridge University Press},
  address       = {Cambridge}
}

@article{MR1682973,
  author        = {Blokhuis, A. and Ball, S. and Brouwer, A. E. and Storme, L.
                  and Sz{\H{o}}nyi, T.},
  title         = {On the number of slopes of the graph of a function defined on
                  a finite field},
  journal       = {J. Combin. Theory Ser. A},
  volume        = {86},
  year          = {1999},
  number        = {1},
  pages         = {187--196},
  mrclass       = {05B25 (51E21)},
  mrnumber      = {MR1682973}
}

@article{MR2019280,
  author        = {Ball, S.},
  title         = {The number of directions determined by a function over a
                  finite field},
  journal       = {J. Combin. Theory Ser. A},
  volume        = {104},
  year          = {2003},
  number        = {2},
  pages         = {341--350},
  mrclass       = {05B25 (12E20)},
  mrnumber      = {MR2019280}
}

@article {MR2375464,
    AUTHOR = {Ball, Simeon},
     TITLE = {On the graph of a function in many variables over a finite field},
   JOURNAL = {Des. Codes Cryptogr.},
  FJOURNAL = {Designs, Codes and Cryptography. An International Journal},
    VOLUME = {47},
      YEAR = {2008},
    NUMBER = {1-3},
     PAGES = {159--164},
      ISSN = {0925-1022},
   MRCLASS = {51E20 (11T06)},
  MRNUMBER = {MR2375464},
}

@article {MR1869411,
    AUTHOR = {Storme, L. and Sziklai, P.},
     TITLE = {Linear point sets and {R}\'edei type $k$-blocking sets in  {${\rm PG}(n,q)$}}}

@article{CRSS98,
  author        = {Calderbank, A. R. and Rains, E. M. and Shor, P. W. and Sloane, N. J. A.},
  title         = {Quantum error correction via codes over {GF}(4)},
  journal       = {IEEE Trans. Inform. Theory},
  fjournal      = {IEEE Transactions on Information Theory},
  volume        = {44},
  year          = {1998},
  number        = {4},
  pages         = {1369--1387},
  issn          = {0018-9448},
  mrclass       = {94B60 (81P68)},
  mrnumber      = {MR1665816},
  doi           = {10.1109/18.681315}
}

@article{AK01,
  author        = {Ashikhmin, A. and Knill, E.},
  title         = {Nonbinary quantum stabilizer codes},
  journal       = {IEEE Trans. Inform. Theory},
  fjournal      = {IEEE Transactions on Information Theory},
  volume        = {47},
  year          = {2001},
  number        = {7},
  pages         = {3065--3072},
  issn          = {0018-9448},
  mrclass       = {94B60 (81P68)},
  mrnumber      = {MR1872852},
  doi           = {10.1109/18.959288}
}

@book{MR0352060,
  author        = {R{\'e}dei, L.},
  title         = {Lacunary polynomials over finite fields},
  note          = {Translated from the German by I. F\"oldes},
  publisher     = {North-Holland Publishing Co.},
  address       = {Amsterdam},
  year          = {1973},
  pages         = {x+257},
  mrclass       = {12C05},
  mrnumber      = {MR0352060}
}

@article{BoseBurton,
  author        = {Bose, R. C. and Burton, R. C.},
  title         = {A characterization of flat spaces in a finite geometry and
                  the uniqueness of the {H}amming and the {M}ac{D}onald codes},
  journal       = {J. Combinatorial Theory},
  volume        = {1},
  year          = {1966},
  pages         = {96--104},
  mrclass       = {50.70 (94.10)},
  mrnumber      = {MR0184873}
}

@article{BGL22,
  author        = {Ball, Simeon and Gamboa, Guillermo and Lavrauw, Michel},
  title         = {On additive {MDS} codes over small fields},
  journal       = {Adv. Math. Commun.},
  fjournal      = {Advances in Mathematics of Communications},
  volume        = {17},
  year          = {2023},
  number        = {4},
  pages         = {828--844},
  doi           = {10.3934/amc.2021024}
}

@article{AdBall,
  author        = {Adriaensen, Sam and Ball, Simeon},
  title         = {On additive {MDS} codes with linear projections},
  journal       = {Finite Fields Appl.},
  fjournal      = {Finite Fields and their Applications},
  volume        = {91},
  year          = {2023},
  pages         = {102255},
  doi           = {10.1016/j.ffa.2023.102255}
}

@article{BBM97,
    AUTHOR = {Ball, S. and Blokhuis, A. and Mazzocca, F.},
     TITLE = {Maximal arcs in {D}esarguesian planes of odd order do not
              exist},
   JOURNAL = {Combinatorica},
  FJOURNAL = {Combinatorica. An International Journal on Combinatorics and
              the Theory of Computing},
    VOLUME = {17},
      YEAR = {1997},
    NUMBER = {1},
     PAGES = {31--41},
      ISSN = {0209-9683},
   MRCLASS = {51E21},
  MRNUMBER = {MR1466573},
}

@article{DDHM02,
    AUTHOR = {De Clerck, F. and Delanote, M. and Hamilton, N. and Mathon, R.},
     TITLE = {Perp-systems and partial geometries},
   JOURNAL = {Adv. Geom.},
  FJOURNAL = {Advances in Geometry},
    VOLUME = {2},
      YEAR = {2002},
    NUMBER = {1},
     PAGES = {1--12},
      ISSN = {1615-715X},
   MRCLASS = {51E14 (05B25)},
}

@misc{AldersonBall26,
    AUTHOR = {Alderson, T. L. and Ball, S.},
     TITLE = {Sets of subspaces with restricted hyperplane intersection
              numbers},
      YEAR = {2026},
      NOTE = {arXiv:2603.27689},
}

@article{MR1772204,
    AUTHOR = {Blokhuis, Aart and Lavrauw, Michel},
     TITLE = {Scattered spaces with respect to a spread in {${\rm PG}(n,q)$}},
   JOURNAL = {Geom. Dedicata},
  FJOURNAL = {Geometriae Dedicata},
    VOLUME = {81},
      YEAR = {2000},
    NUMBER = {1-3},
     PAGES = {231--243},
      ISSN = {0046-5755},
   MRCLASS = {51E20},
  MRNUMBER = {MR1772204},
}

@article{MR2684078,
    AUTHOR = {Polverino, Olga},
     TITLE = {Linear sets in finite projective spaces},
   JOURNAL = {Discrete Math.},
  FJOURNAL = {Discrete Mathematics},
    VOLUME = {310},
      YEAR = {2010},
    NUMBER = {22},
     PAGES = {3096--3107},
      ISSN = {0012-365X},
   MRCLASS = {51E20},
  MRNUMBER = {MR2684078},
}

@article{LMPT14,
    AUTHOR = {Lunardon, Guglielmo and Marino, Giuseppe and Polverino, Olga
              and Trombetti, Rocco},
     TITLE = {Maximum scattered linear sets of pseudoregulus type and the
              {S}egre variety {$\mathcal{S}_{n,n}$}},
   JOURNAL = {J. Algebraic Combin.},
  FJOURNAL = {Journal of Algebraic Combinatorics},
    VOLUME = {39},
      YEAR = {2014},
    NUMBER = {4},
     PAGES = {807--831},
}

@article{SzT12,
    AUTHOR = {Sziklai, P{\'e}ter and Tak{\'a}ts, Marcella},
     TITLE = {An extension of the direction problem},
   JOURNAL = {Discrete Math.},
  FJOURNAL = {Discrete Mathematics},
    VOLUME = {312},
      YEAR = {2012},
    NUMBER = {12-13},
     PAGES = {2083--2087},
       DOI = {10.1016/j.disc.2012.02.021},
}

@article{KKKS06,
  author        = {Ketkar, A. and Klappenecker, A. and Kumar, S. and Sarvepalli, P. K.},
  title         = {Nonbinary stabilizer codes over finite fields},
  journal       = {IEEE Trans. Inform. Theory},
  fjournal      = {IEEE Transactions on Information Theory},
  volume        = {52},
  year          = {2006},
  number        = {11},
  pages         = {4892--4914},
  issn          = {0018-9448},
  mrclass       = {94B60},
  mrnumber      = {MR2300823},
  doi           = {10.1109/TIT.2006.883612}
}

@article{BLP24,
  author        = {Ball, Simeon and Lavrauw, Michel and Popatia, Tabriz},
  title         = {Griesmer type bounds for additive codes over finite fields,
                  integral and fractional {MDS} codes},
  journal       = {Des. Codes Cryptogr.},
  fjournal      = {Designs, Codes and Cryptography},
  volume        = {93},
  year          = {2025},
  pages         = {175--196},
}

\end{document}